\documentclass{article}
\usepackage[margin=1.5in]{geometry}

\usepackage{amsmath,amssymb,amsfonts,amsthm,xcolor}
\usepackage{graphicx,float}
\usepackage{paralist,hyperref}
\usepackage{caption,subcaption}
\usepackage{siunitx,tabularx}

\usepackage{dcolumn,booktabs}
\usepackage[round]{natbib}
\title{Backfitting for large scale crossed random effects
  regressions}
\author{Swarnadip Ghosh\\Stanford University
  \and Trevor Hastie\\Stanford University
\and Art B. Owen \\Stanford University}
\date{February 2021}

\allowdisplaybreaks

\newtheorem{lemma}{Lemma}
\newtheorem{theorem}{Theorem}

\newcommand{\myemph}[1]{\textsl{\textbf{#1}}}

\newcommand{\bone}{\mathbf{1}}

\newcommand{\tran}{\mathsf{T}}
\newcommand{\wh}{\widehat}
\newcommand{\wt}{\widetilde}
\newcommand{\diag}{\mathrm{diag}}

\renewcommand{\ge}{\geqslant}
\renewcommand{\le}{\leqslant}

\newcommand{\dnorm}{\mathcal{N}}
\newcommand{\dunif}{\mathbb{U}}

\newcommand{\sumdot}{\text{\tiny$\bullet$}}
\newcommand{\e}{\mathbb{E}}

\newcommand{\zij}{Z_{ij}}
\newcommand{\zijp}{Z_{ij'}}
\newcommand{\zis}{Z_{is}}

\newcommand{\zipj}{Z_{i'j}}
\newcommand{\zipjp}{Z_{i'j'}}

\newcommand{\zrj}{Z_{rj}}

\newcommand{\zrs}{Z_{rs}}

\newcommand{\zzt}{(ZZ^\tran)}
\newcommand{\ztz}{(Z^\tran Z)}

\newcommand{\nid}{N_{i\sumdot}}
\newcommand{\ndj}{N_{\sumdot j}}

\newcommand{\nrd}{N_{r\sumdot}}
\newcommand{\nds}{N_{\sumdot s}}

\newcommand{\ai}{a_i}

\newcommand{\bj}{b_{j}}

\newcommand{\bs}{b_s}

\newcommand{\eij}{e_{ij}}

\newcommand{\ssa}{\sigma^2_A}
\newcommand{\ssb}{\sigma^2_B}
\newcommand{\sse}{\sigma^2_E}

\newcommand{\ssah}{\hat\sigma^2_A}
\newcommand{\ssbh}{\hat\sigma^2_B}
\newcommand{\sseh}{\hat\sigma^2_E}

\newcommand{\yij}{Y_{ij}}

\newcommand{\yrj}{Y_{rj}}

\newcommand{\yis}{Y_{is}}

\newcommand{\yrs}{Y_{rs}}

\newcommand{\xij}{x_{ij}}

\newcommand{\var}{\mathrm{var}}    
\newcommand{\cov}{\mathrm{cov}}   

\newcommand{\simiid}{\stackrel{\mathrm{iid}}{\sim}}
\newcommand{\simind}{\stackrel{\mathrm{ind}}{\sim}}

\newcommand{\real}{\mathbb{R}}

\newcommand{\phe}{\phantom{=}\,\,} 
\newcommand{\phz}{\phantom{0}}

\newcommand{\ols}{\mathrm{OLS}}
\newcommand{\gls}{\mathrm{GLS}}

\newcommand{\bsa}{\boldsymbol{a}}
\newcommand{\bsb}{\boldsymbol{b}}
\newcommand{\bsg}{\boldsymbol{g}}
\newcommand{\bsu}{\boldsymbol{u}}
\newcommand{\bsx}{\boldsymbol{x}}
\newcommand{\bszero}{\boldsymbol{0}}
\newcommand{\bgamma}{\boldsymbol{\gamma}}
\newcommand{\dbern}{\mathrm{Bern}}
\newcommand{\dbin}{\mathrm{Bin}}

\newcommand{\pij}{p_{ij}}

\newcommand{\stocleq}{\preccurlyeq}

\newcommand{\pup}{\Upsilon}

\newcommand{\mzero}{M^{(0)}}
\newcommand{\mone}{M^{(1)}}
\newcommand{\mtwo}{M^{(2)}}
\newcommand{\mthree}{M^{(3)}}

\newcommand{\cD}{\mathcal{D}}
\newcommand{\cf}{f}

\newcommand{\cS}{\mathcal{S}}
\newcommand{\cR}{\mathcal{R}}
\newcommand{\cG}{\mathcal{G}}

\newcommand{\cu}{\mathcal{U}}
\newcommand{\cv}{\mathcal{V}}
\newcommand{\cw}{\mathcal{W}}
\newcommand{\cx}{\mathcal{X}}
\newcommand{\cy}{\mathcal{Y}}
\newcommand{\cz}{\mathcal{Z}}
\newcommand{\bse}{\boldsymbol{e}}

\newcommand{\dom}{\mathcal{D}}
\newcommand{\rhoaux}{\rho_{\mathrm{aux}}}
\newcommand{\rhoprz}{\rho_{\mathrm{PRZ}}}

\begin{document}
\maketitle
\begin{abstract}
Regression models with crossed random effect errors can be very expensive to compute.
The cost of both generalized least squares and Gibbs sampling
can easily grow as $N^{3/2}$ (or worse) for $N$ observations.
\cite{papa:robe:zane:2020} present a collapsed Gibbs sampler that costs
$O(N)$, but under an extremely stringent sampling model.
We propose a backfitting algorithm to compute a generalized least squares
estimate and prove that it costs $O(N)$.
A critical part of the proof is in ensuring that the number
of iterations required is $O(1)$ which follows from keeping
a certain matrix norm below $1-\delta$ for some $\delta>0$.
Our conditions are greatly relaxed compared to those for the collapsed Gibbs
sampler, though still strict.
Empirically, the backfitting algorithm
has a norm below $1-\delta$ under conditions that are less strict
than those in our assumptions.
We illustrate the new algorithm on a ratings data set from Stitch Fix.
\end{abstract}

\section{Introduction}
To estimate a regression when the
errors have a non-identity covariance matrix, we
usually turn first to generalized least squares (GLS).  Somewhat
surprisingly, GLS proves to be computationally challenging
in the very simple setting of the unbalanced crossed random
effects models that we study here.
For that problem, the cost to compute the GLS estimate on $N$
data points grows at best like $O(N^{3/2})$ under the usual algorithms.
If we additionally assume Gaussian errors, then
\cite{crelin} show that even evaluating
the likelihood one time costs at least a multiple of $N^{3/2}$.
These costs make the usual algorithms for GLS
infeasible for large data sets such as
those arising in electronic commerce.

In this paper, we present an iterative algorithm based
on a backfitting approach from \cite{buja:hast:tibs:1989}.
This algorithm is known to converge to the
GLS solution.  The cost of each iteration is $O(N)$
and so we also study how the number of iterations grows
with~$N$.

The crossed random effects model we consider has
\begin{equation}\label{eq:refmodel}
\yij =\xij^\tran\beta+a_i+b_j+e_{ij},\quad 1\le i\le R,\quad
1\le j\le C
\end{equation}
for random effects $\ai$ and $\bj$ and an error $\eij$
with a fixed effects regression parameter $\beta\in\real^p$ for the covariates $\xij\in\real^p$.
We assume that
$a_i\simiid (0,\ssa)$, $b_j\simiid(0,\ssb)$, and $\eij\simiid(0,\sse)$
are all independent.  It is thus a mixed effects model in which the
random portion has a crossed structure.
The GLS estimate is also the maximum likelihood
estimate (MLE), when $\ai$, $\bj$ and $\eij$ are Gaussian.
Because we assume that $p$ is fixed as $N$ grows, we often
leave $p$ out of our cost estimates, giving instead the complexity in $N$.

The GLS estimate $\hat\beta_\gls$ for crossed random effects can be efficiently
estimated if all $R\times C$ values are available.
Our motivating examples involve ratings data where $R$ people
rate $C$ items and then it is usual that the data are
very unbalanced with a haphazard observational pattern
in which only $N\ll R\times C$ of the $(\xij,\yij)$ pairs are observed.
The crossed random effects setting is significantly more difficult
than a hierarchical model with just $\ai+\eij$ but no $\bj$
term.  Then the observations for index $j$ are `nested within' those
for each level of index $i$. The result is that the covariance matrix
of all observed $\yij$ values has a block diagonal structure
allowing GLS to be computed in $O(N)$ time.

Hierarchical models are very well suited to Bayesian
computation \citep{gelm:hill:2006}.
Crossed random effects are a much greater challenge.
\cite{GO17} find that the Gibbs sampler can take $O(N^{1/2})$
iterations to converge to stationarity, with each iteration
costing $O(N)$ leading once again to $O(N^{3/2})$ cost.
For more examples where 
the costs of solving equations versus sampling from a
covariance attain the same rate see
\cite{good:soka:1989} and \cite{RS97}.
As further evidence of the difficulty of this problem,
the Gibbs sampler was one of nine MCMC algorithms that
\cite{GO17} found to be unsatisfactory.
Furthermore, \cite{lme4} removed the {\tt mcmcsamp} function from the R package lme4
because it was considered unreliable even for the problem
of sampling the posterior distribution of the parameters
from previously fitted models, and even for those with random effects variances
near zero.

\cite{papa:robe:zane:2020} present
an exception to the high cost of a Bayesian approach
for crossed random effects.  They propose a collapsed Gibbs
sampler that can potentially mix in $O(1)$ iterations.
To prove this rate, they make an extremely stringent
assumption that every index $i=1,\dots,R$ appears in the
same number $N/C$ of observed data points and similarly
every $j=1,\dots,C$ appears in $N/R$ data points.
Such a condition is tantamount to requiring a designed
experiment for the data and it is much stronger than
what their algorithm seems to need in practice.
Under that condition their mixing rate asymptotes
to a quantity $\rhoaux$, described in our discussion section, 
that in favorable circumstances is $O(1)$.
They find empirically that their sampler has a cost that
scales well in many data sets where their balance condition
does not hold.

In this paper we study an iterative linear operation,
known as backfitting, for GLS.
Each iteration costs $O(N)$.
The speed of convergence depends on a certain
matrix norm of that iteration, which we exhibit below.
If the norm remains bounded strictly below $1$
as $N\to\infty$, then
the number of iterations to convergence is $O(1)$.
We are able to show that the matrix norm is $O(1)$
with probability tending to one, under conditions where
the number of observations per row (or per column) is random
and even the expected row or column counts may vary,
though in a narrow range.
While this is a substantial weakening of the conditions in
\cite{papa:robe:zane:2020}, it still fails to cover many
interesting cases.  Like them, we find empirically that our
algorithm scales much more broadly than under the
conditions for which scaling is proved.

We suspect that the computational infeasibility of GLS leads
many users to  use ordinary least squares (OLS) instead.
OLS has two severe problems.
First, it is \myemph{inefficient} with
$\var(\hat\beta_\ols)$ larger than $\var(\hat\beta_\gls)$.
This is equivalent to OLS ignoring some possibly large
fraction of the information in the data.
Perhaps more seriously, OLS is \myemph{naive}.
It produces an estimate of $\var(\hat\beta_\ols)$ that
can be too small by a large factor.  That amounts
to overestimating the quantity of information behind $\hat\beta_\ols$,
also by a potentially large factor.

The naivete of OLS can be countered by using better variance estimates.
One can bootstrap it by resampling the row and column entities as in \cite{pbs}.
There is also a version of Huber-White variance estimation
for this case in econometrics. See for instance \cite{came:gelb:mill:2011}.
While these methods counter the naivete of OLS, the inefficiency of OLS remains.

The method of moments algorithm in \cite{crelin}
gets consistent asymptotically normal estimates
of $\beta$, $\ssa$, $\ssb$ and $\sse$.
It produces a GLS estimate $\hat\beta$ that is more
efficient than OLS but still not fully efficient
because it accounts for correlations due to only one of the
two crossed random effects.  While inefficient, it is not naive
because its estimate of $\var(\hat\beta)$
properly accounts for variance due to $\ai$, $\bj$ and $\eij$.

In this paper we get a GLS estimate $\hat\beta$
that takes account of all three variance components,
making it efficient.
We also provide an estimate of $\var(\hat\beta)$ that accounts
for all three components, so our estimate is not naive.
Our algorithm requires consistent estimates of the variance components
$\ssa$, $\ssb$ and $\sse$ in computing $\hat\beta$ and $\wh\var(\hat\beta)$.
We use the method of moments estimators from \cite{GO17} that can
be computed in $O(N)$ work.
By \citet[Theorem 4.2]{GO17}, these estimates of $\ssa$, $\ssb$ and
$\sse$ are asymptotically uncorrelated and each of them has the same
asymptotic variance it would have had were the other two variance components equal to zero.
It is not known whether they are optimally estimated, much less optimal
subject to an $O(N)$ cost constraint.
The variance component estimates are known to be
asymptotically normal \citep{gao:thesis}.

The rest of this paper is organized as follows.
Section~\ref{sec:missing} introduces our notation and assumptions
for missing data.
Section~\ref{sec:backfitting} presents the backfitting algorithm
from \cite{buja:hast:tibs:1989}.  That algorithm was defined for
smoothers, but we are able to cast the estimation of random effect
parameters as a special kind of smoother.
Section~\ref{sec:normconvergence} proves our result about
backfitting being convergent with a probability tending to one
as the problem size increases.
Section~\ref{sec:empiricalnorms} shows numerical measures
of the matrix norm of the backfitting operator.  It remains
bounded below and away from one under more conditions than our theory shows.
We find that even one iteration of
the lmer function in lme4 package \cite{lme4} has a cost that grows like $N^{3/2}$
in one setting and like $N^{2.1}$ in another, sparser one.
The backfitting algorithm has cost $O(N)$ in both of these cases.
Section~\ref{sec:stitch} illustrates our GLS algorithm
on some data provided to us by Stitch Fix.  These are customer
ratings of items of clothing on a ten point scale.
Section~\ref{sec:discussion} has a discussion of these results.
An appendix contains some regression output for the
Stitch Fix data.

\section{Missingness}\label{sec:missing}

We adopt the notation from \cite{crelin}.
We let $\zij\in\{0,1\}$ take the value $1$
if $(\xij,\yij)$ is observed  and $0$ otherwise,
for $i=1,\dots,R$ and $j=1,\dots,C$.
In many of the contexts we consider, the missingness
is not at random and is potentially informative.
Handling such problems is outside the scope of
this paper, apart from a brief discussion in Section~\ref{sec:discussion}.
It is already a sufficient challenge to work without
informative missingness.

The matrix $Z\in\{0,1\}^{R\times C}$, with elements $\zij$
has $\nid =\sum_{j=1}^C\zij$ observations
in `row $i$' and $\ndj=\sum_{i=1}^R\zij$ observations
in `column $j$'.
We often drop the limits of summation so that $i$
is always summed over $1,\dots,R$ and $j$ over $1,\dots,C$.
When we need additional symbols for row and column indices we
use $r$ for rows and $s$ for columns.
The total sample size is $N=\sum_i\sum_j\zij
=\sum_i\nid = \sum_j\ndj$.

There are two co-observation matrices, $Z^\tran Z$ and $ZZ^\tran$.
Here $\ztz_{js}=\sum_i\zij\zis$ gives the number of rows in which
data from both columns $j$ and $s$ were observed,
while $\zzt_{ir}=\sum_j\zij\zrj$ gives the number of
columns in which data from both rows $i$ and $r$ were observed.

In our regression models, we treat $\zij$ as nonrandom. We are conditioning
on the actual pattern of observations in our data.
When we study the rate at which our backfitting algorithm converges, we
consider $\zij$ drawn at random. That is, the analyst is solving a GLS
conditionally on the pattern of observations and missingness, while
we study the convergence rates that analyst will see for data
drawn from a missingness mechanism defined in Section~\ref{sec:modelz}.

If we place all of the $\yij$ into a vector $\cy\in\real^N$ and $\xij$
compatibly into a matrix $\cx\in\real^{N\times p}$, then
the naive and inefficient OLS estimator is
\begin{align}\label{eq:bhatols}
\hat\beta_\ols = (\cx^\tran \cx)^{-1}\cx^\tran\cy.
\end{align}
This can be computed in $O(Np^2)$ work.  We prefer to use
the GLS estimator
\begin{align}\label{eq:bhatgls}\hat\beta_\gls = (\cx^\tran \cv^{-1}\cx)^{-1}\cx^\tran\cv^{-1}\cy,
\end{align}
where $\cv\in\real^{N\times N}$ contains all of the $\cov(\yij,\yrs)$ in
an ordering compatible with $\cx$ and $\cy$.  A naive algorithm costs $O(N^3)$
to solve for $\hat\beta_\gls$.
It can actually be solved through a Cholesky decomposition of an $(R+C)\times (R+C)$ matrix
\citep{sear:case:mccu:1992}.
That has cost $O(R^3+C^3)$.
Now $N\le RC$, with equality only for completely observed data.
Therefore  $\max(R,C)\ge \sqrt{N}$, and so $R^3+C^3\ge N^{3/2}$.
When the data are sparsely enough observed it is possible
that $\min(R,C)$ grows more rapidly than $N^{1/2}$.
In a numerical example in Section~\ref{sec:empiricalnorms} we have $\min(R,C)$
growing like $N^{0.70}$.
In a hierarchical model, with $\ai$ but no $\bj$ we would find
$\cv$ to be block diagonal and then
$\hat\beta_\gls$ could be computed in $O(N)$ work.

A reviewer reminds us that it has been known since \cite{stra:1969} that
systems of equations can be solved more quickly than cubic time.
Despite that, current software is still dominated by cubic time algorithms.
Also none of the known solutions are quadratic
and so in our setting the cost would be at least a multiple
of $(R+C)^{2+\gamma}$ for some $\gamma>0$ and hence not $O(N)$.

We can write our crossed effects model as
\begin{align}\label{eq:cemodelviaz}
\cy = \cx\beta + \cz_A\bsa + \cz_B\bsb + \bse
\end{align}
for matrices $\cz_A\in\{0,1\}^{N\times R}$ and $\cz_B\in\{0,1\}^{N\times C}$.
The $i$'th column of $\cz_A$ has ones for all of the $N$ observations that
come from row $i$ and zeroes elsewhere. The definition of $\cz_B$ is analogous.
The observation matrix can be written $Z = \cz_A^\tran\cz_B$.
The vector $\bse$ has all $N$ values of $\eij$ in compatible order.
Vectors $\bsa$ and $\bsb$ contain the row and column random effects
$\ai$ and $\bj$.
In this notation
\begin{equation}
  \label{eq:Vee}
  \cv = \cz_A\cz_A^\tran\ssa + \cz_B\cz_B^\tran\ssb + I_N\sse,
\end{equation}
where $I_N$ is the $N \times N$ identity matrix.

Our main computational problem is to get
a value for $\cu=\cv^{-1}\cx\in\real^{N\times p}$.
To do that we iterate towards a solution $\bsu\in\real^N$ of  $\cv \bsu=\bsx$,
where $\bsx\in\real^N$ is one of the $p$ columns of $\cx$.
After that, finding
\begin{equation}
  \label{eq:betahat}
\hat\beta_\gls = (\cx^\tran \cu)^{-1}(\cy^\tran\cu)^\tran
\end{equation}
is not expensive, because $\cx^\tran\cu\in\real^{p\times p}$ and we suppose that $p$ is not large.

If the data ordering in $\cy$ and elsewhere sorts by index $i$, breaking ties by index $j$,
then $\cz_A\cz_A^\tran\in\{0,1\}^{N\times N}$ is
a block matrix with $R$ blocks of ones
of size $\nid\times\nid$ along the diagonal and zeroes elsewhere.
The matrix $\cz_B\cz_B^\tran$ will not be block diagonal in that ordering.
Instead $P\cz_B\cz_B^\tran P^\tran$ will be block diagonal with
$\ndj\times\ndj$ blocks of ones on the diagonal,
for a suitable $N\times N$ permutation matrix $P$.

\section{Backfitting algorithms}\label{sec:backfitting}
Our first goal is to develop computationally efficient ways to
solve the GLS problem \eqref{eq:betahat} for the linear mixed model~\eqref{eq:cemodelviaz}.
We use the  backfitting algorithm that
\cite{hast:tibs:1990} and \cite{buja:hast:tibs:1989}
use to fit additive models.
We write $\cv$ in (\ref{eq:Vee}) as $\sse\left(\cz_A\cz_A^\tran/\lambda_A+\cz_B\cz_B^\tran/\lambda_B
  +I_N\right)$ with $\lambda_A=\sse/\ssa$ and
$\lambda_B=\sse/\ssb$,
and define $\cw=\sse\cv^{-1}$.  
Then the GLS estimate of $\beta$ is
\begin{align}
  \hat\beta_{\gls}&=\arg\min_\beta (\cy-\cx\beta)^\tran\cw(\cy-\cx\beta)
=     (\cx^\tran\cw\cx)^{-1}\cx^\tran\cw\cy\label{eq:betahatw}
\end{align}
and $\cov(\hat\beta_{\gls})=\sse (\cx^\tran\cw\cx)^{-1}$.

It is well known (e.g., \cite{robinson91:_that_blup}) that we can obtain
$\hat\beta_{\gls}$ by solving the
following penalized least-squares problem
\begin{align}\label{eq:minboth}
\min_{\beta,\bsa,\bsb}\Vert \cy-\cx\beta-\cz_A\bsa-\cz_B\bsb\Vert^2
+\lambda_A\Vert\bsa\Vert^2 +\lambda_B\Vert\bsb\Vert^2.
\end{align}
Then $\hat\beta=\hat\beta_{\gls}$ and $\hat \bsa$ and $\hat \bsb$ are the
best linear unbiased prediction (BLUP) estimates
of the random effects.
This derivation works for any number of factors, but it is
instructive to carry it through initially for one.


\subsection{One factor}\label{sec:one-factor}

For a single factor,
we simply drop the $\cz_B\bsb$ term from \eqref{eq:cemodelviaz} to get
\begin{equation*}
  \cy = \cx\beta + \cz_A\bsa +\bse.
\end{equation*}
Then
$\cv=\cov(\cz_A\bsa+\bse)= \ssa\cz_A\cz_A^\tran +\sse I_N$, and $\cw=\sse\cv^{-1}$ as before.
The penalized least squares problem is to solve
\begin{align}\label{eq:equivmina}
\min_{\beta,\bsa} \Vert \cy - \cx\beta -\cz_A\bsa\Vert^2 + \lambda_A \Vert\bsa\Vert^2.
\end{align}
We show the details as we need them for a later derivation.

The normal equations from~\eqref{eq:equivmina} yield
\begin{align}
\bszero & = \cx^\tran(\cy-\cx\hat\beta-\cz_A\hat\bsa),\quad\text{and}\label{eq:normbeta}\\
\bszero & = \cz_A^\tran(\cy-\cx\hat\beta-\cz_A\hat\bsa)
    -\lambda_A\hat\bsa.\label{eq:normbsa}
\end{align}
Solving~\eqref{eq:normbsa} for $\hat\bsa$ and multiplying the solution by $\cz_A$ yields
$$
\cz_A\hat\bsa = \cz_A(\cz_A^\tran \cz_A + \lambda_AI_R)^{-1}\cz_A^\tran(\cy-\cx\hat\beta)
\equiv \cS_A(\cy-\cx\hat\beta),
$$
for an $N\times N$ ridge regression ``smoother matrix'' $\cS_A$.
As we explain below this smoother matrix implements shrunken within-group means.
Then substituting $\cz_A\hat\bsa$ into equation~\eqref{eq:normbeta}    
yields
\begin{equation}
  \label{eq:onefactor}
\hat\beta = (\cx^\tran(I_N-\cS_A)\cx)^{-1}\cx^\tran(I_N-\cS_A)\cy.
\end{equation}
Using the Sherman-Morrison-Woodbury (SMW) identity, one can show that $\cw=I_N-\cS_A$ and
hence $\hat\beta$ above equals $\hat\beta_\gls$
from~\eqref{eq:betahatw}. This is not in itself a new discovery; see
for example \cite{robinson91:_that_blup} or \cite{hast:tibs:1990}
(Section 5.3.3).

To compute the solution in (\ref{eq:onefactor}), we need to compute
$\cS_A \cy$  and $\cS_A\cx$. The heart of the computation in
$\cS_A \cy$
is $(\cz_A^\tran \cz_A + \lambda_AI_R)^{-1}\cz_A^\tran\cy$.
But $\cz_A^\tran
\cz_A=\diag(N_{1\sumdot},N_{2\sumdot},\ldots,N_{R\sumdot})$ and we
see that all we are doing is computing an $R$-vector of shrunken means of the elements
of $\cy$ at each level of the factor $A$; the $i$th element is $\sum_j\zij Y_{ij}/(N_{i\sumdot}+\lambda_A)$.
This involves a single pass through the $N$ elements of $Y$,
accumulating the sums into $R$ registers, followed by an elementwise
scaling of the $R$ components.  Then pre-multiplication by $\cz_A$ simply puts these
$R$ shrunken means back into an
$N$-vector in the appropriate positions. The total cost is $O(N)$.
Likewise $\cS_A\cx$ does the
same separately for each of the columns of $\cx$.
Hence the entire computational cost for \eqref{eq:onefactor} is $O(Np^2)$, the same order as regression on $\cx$.

What is also clear is that the indicator matrix
$\cz_A$ is not actually needed here; instead all we need to carry out
these computations is the
factor vector $\cf_A$ that records the level of factor $A$ for each
of the $N$ observations. In the R language \citep{R:lang:2015} the following pair of operations does
the computation:
\begin{verbatim}
  hat_a = tapply(y,fA,sum)/(table(fA)+lambdaA)
  hat_y = hat_a[fA]
\end{verbatim}
where {\tt fA} is a categorical variable (factor) $f_A$ of length $N$ containing the row indices $i$ in an order compatible with $Y\in\real^N$ (represented as {\tt y})
and {\tt lambdaA} is $\lambda_A=\ssa/\sse$.

\subsection{Two factors}\label{sec:two-factors}
With two factors we face the problem of incompatible block diagonal
matrices discussed in Section~\ref{sec:missing}.
Define $\cz_G=(\cz_A\!:\!\cz_B)$ ($R+C$ columns),
$\cD_\lambda=\diag(\lambda_AI_R,\lambda_BI_C)$,
and $\bsg^\tran=(\bsa^\tran,\bsb^\tran)$.
Then solving \eqref{eq:minboth} is equivalent to
\begin{align}\label{eq:ming}
\min_{\beta,\bsg}\Vert \cy-\cx\beta-\cz_G\bsg\Vert^2
+\bsg^\tran\cD_\lambda\bsg.
\end{align}

A derivation similar to that used in the one-factor case gives
\begin{equation}
  \label{eq:gfactor}
\hat\beta =
H_\gls\cy\quad\text{for}\quad
H_\gls = (\cx^\tran(I_N-\cS_G)\cx)^{-1}\cx^\tran(I_N-\cS_G),
\end{equation}
where the hat matrix $H_\gls$ is written in terms of
a smoother matrix
\begin{equation}
  \label{eq:defcsg}
  \cS_G=\cz_G(\cz_G^\tran \cz_G + \cD_\lambda)^{-1}\cz_G^\tran.
\end{equation}
We can again use SMW to show that $I_N-\cS_G=\cw$ and hence the
solution $\hat\beta=\hat\beta_{\gls}$ in \eqref{eq:betahatw}.
But in applying $\cS_G$ we do not enjoy the computational
simplifications that occurred in the one factor case, because
\begin{equation*}
  \cz_G^\tran\cz_G=
  \left(
    \begin{array}{cc}
    \cz_A^\tran\cz_A&\cz_A^\tran\cz_B\\[0.25ex]
      \cz_B^\tran\cz_A&\cz_B^\tran\cz_B
    \end{array}
  \right)
=\begin{pmatrix} \diag(\nid) & Z\\
Z^\tran & \diag(\ndj)
\end{pmatrix},
\end{equation*}
where $Z\in\{0,1\}^{R\times C}$ is the observation matrix
which has no special structure.
Therefore we need to invert an $(R+C)\times (R+C)$ matrix to apply
$\cS_G$ and hence to solve
\eqref{eq:gfactor}, at a cost of at least $O(N^{3/2})$ (see Section~\ref{sec:missing}).

Rather than group $\cz_A$ and $\cz_B$, we keep them separate, and
develop an algorithm to apply the operator $\cS_G$ efficiently.
Consider a generic response vector $\cR$ (such as $\cy$ or a column of $\cx$) and the optimization problem
\begin{align}\label{eq:minab}
\min_{\bsa,\bsb}\Vert \cR-\cz_A\bsa-\cz_B\bsb\Vert^2
+\lambda_A\|\bsa\|^2+\lambda_B\|\bsb\|^2.
\end{align}

Using  $\cS_G$ defined at~\eqref{eq:defcsg}
in terms of the indicator variables $\cz_G\in\{0,1\}^{N\times (R+C)}$
it is clear that the fitted values are given by
  $\widehat\cR=\cS_G\cR$.
Solving (\ref{eq:minab}) would result in two blocks of estimating
equations similar to equations \eqref{eq:normbeta} and \eqref{eq:normbsa}.
These can be written
\begin{align}\label{eq:backfit}
\begin{split}
\cz_A\hat\bsa & = \cS_A(\cR-\cz_B\hat\bsb),\quad\text{and}\\
\cz_B\hat\bsb & = \cS_B(\cR-\cz_A\hat\bsa),
\end{split}
\end{align}
where
$\cS_A=\cz_A(\cz_A^\tran\cz_A + \lambda_AI_R)^{-1}\cz_A^\tran$ is again
the ridge regression smoothing matrix for row effects and similarly
$\cS_B=\cz_B(\cz_B^\tran\cz_B + \lambda_BI_C)^{-1}\cz_B^\tran$ the
smoothing matrix for column effects.
We solve these equations iteratively by block coordinate descent,
also known as backfitting.
The iterations converge to the solution
of~\eqref{eq:minab} \citep{buja:hast:tibs:1989, hast:tibs:1990}.

It is evident that $\cS_A,\cS_B\in\real^{N\times N}$
are both symmetric matrices.  It follows that the limiting smoother
$\cS_G$ formed by combining them is also symmetric. See \citet[page 120]{hast:tibs:1990}.
We will need this result later for an important computational shortcut.

Here the simplifications we enjoyed in the one-factor case once again
apply. Each step applies its operator to a vector
(the terms  in parentheses on the right hand side in
(\ref{eq:backfit})). For both $\cS_A$ and $\cS_B$ these are
simply the shrunken-mean operations described for the one-factor case,
separately for factor $A$ and $B$ each time. As before, we do not need to
actually construct $\cz_B$, but simply use a factor $\cf_B$
that records the level of factor $B$ for each of the $N$ observations.

The above description holds for a generic response $\cR$; we apply that algorithm (in
parallel) to $\cy$ and each column of $\cx$ to obtain
the quantities $\cS_G\cx$ and $\cS_G\cy$
that we need to compute $H_{\gls}\cy$ in \eqref{eq:gfactor}.
Now solving (\ref{eq:gfactor}) is $O(Np^2)$ plus a negligible $O(p^3)$ cost.
These computations deliver $\hat\beta_{\gls}$; if the BLUP
estimates $\hat\bsa$ and $\hat{\bsb}$ are also required, the same algorithm
can be applied to the response $\cy-\cx\hat\beta_{\gls}$, retaining the $\bsa$ and
$\bsb$ at the final iteration.
We can also write
\begin{equation}\label{eq:covbhat}
\cov(\hat\beta_{\gls})=\sse(\cx^\tran(I_N-\cS_G)\cx)^{-1}.
\end{equation}
It is also clear that we can trivially extend this approach to
accommodate any number of factors.

\subsection{Centered operators}
\label{sec:centered-operators}
The matrices $\cz_A$ and $\cz_B$ both have row sums all ones, since
they are factor indicator matrices (``one-hot encoders''). This
creates a nontrivial intersection between their column spaces, and
that of $\cx$ since we always include an intercept, that can
cause backfitting to converge more slowly. In this section we show
how to counter this intersection of column spaces
to speed convergence.
We work with this two-factor model
\begin{align}\label{eq:equivmina1}
\min_{\beta,\bsa,\bsb} \Vert \cy - \cx\beta -\cz_A\bsa-\cz_B\bsb\Vert^2 + \lambda_A \Vert\bsa\Vert^2+\lambda_B\Vert\bsb\Vert^2.
\end{align}
\begin{lemma}
  If $\cx$ in model~\eqref{eq:equivmina1}
includes a column of ones (intercept), and $\lambda_A>0$
  and $\lambda_B>0$, then the solutions for $\bsa$ and $\bsb$ satisfy
  $\sum_{i=1}^R a_i=0$ and $\sum_{j=1}^C b_j=0$.
\end{lemma}
\begin{proof}
It suffices to show this for one factor and with $\cx=\bone$. The
objective is now
\begin{align}\label{eq:equivsimp}
\min_{\beta,\bsa} \Vert \cy - \bone\beta -\cz_A\bsa\Vert^2 + \lambda_A \Vert\bsa\Vert^2.
\end{align}
Notice that for any candidate solution $(\beta,\{a_i\}_1^R)$, the alternative
solution $(\beta+c,\{a_i-c\}_1^R)$ leaves the loss part of
\eqref{eq:equivsimp} unchanged, since the row sums of $\cz_A$ are all
one. Hence if $\lambda_A>0$, we would always improve $\bsa$ by picking
 $c$ to minimize the
penalty term
$\sum_{i=1}^R(a_i-c)^2$, or $c=(1/R)\sum_{i=1}^Ra_i$.
\end{proof}

It is natural then to solve for $\bsa$ and $\bsb$ with these
constraints enforced, instead of waiting for them
to simply emerge in the process of iteration.

\begin{theorem}\label{thm:smartcenter}
  Consider the generic optimization problem
  \begin{align}\label{eq:equivsimp2}
\min_{\bsa} \Vert \cR  -\cz_A\bsa\Vert^2 + \lambda_A
    \Vert\bsa\Vert^2\quad \mbox{subject to } \sum_{i=1}^Ra_i=0.
  \end{align}
  Define the partial sum vector $\cR^+ = \cz_A^\tran\cR$
with components $\cR^+_{i} = \sum_j\zij\cR_{ij}$,
and let
    $$w_i=\frac{(N_{i\sumdot}+\lambda)^{-1}}{\sum_{r}(\nrd+\lambda)^{-1}}.$$
    Then the solution $\hat \bsa$ is given by
\begin{align}\label{eq:ahatsoln}
\hat
    a_i=\frac{\cR^+_{i}-\sum_{r}w_r\cR^+_{r}}{N_{i\sumdot}+\lambda_A},
\quad i=1,\ldots,R.
\end{align}
    Moreover, the fit is given by
    $$\cz_A\hat\bsa=\tilde\cS_A\cR,$$ where $\tilde \cS_A$ is a
    symmetric operator.
\end{theorem}
The computations are a simple modification of the non-centered case.
\begin{proof}
  Let $M$ be an $R\times R$ orthogonal matrix with first column
  $\bone/\sqrt{R}$. Then $\cz_A\bsa=\cz_AMM^\tran\bsa=\tilde
  \cG\tilde\bgamma$ for $\cG=\cz_AM$ and
$\tilde\bgamma=M^\tran\bsa$.
Reparametrizing in this way leads to
  the equivalent problem
  \begin{align}\label{eq:equivsimp2}
\min_{\tilde\bgamma} \Vert \cR  -\tilde\cG\tilde\bgamma\Vert^2 + \lambda_A
    \Vert\tilde\bgamma\Vert^2,\quad \mbox{subject to } \tilde\gamma_1=0.
  \end{align}
To solve (\ref{eq:equivsimp2}), we simply drop the first column of
$\tilde \cG$. Let $\cG=\cz_AQ$ where $Q$ is the matrix $M$ omitting
the first column, and $\bgamma$ the corresponding subvector of
$\tilde\bgamma$ having $R-1$ components. We now solve
  \begin{align}\label{eq:equivsimp3}
\min_{\tilde\bgamma} \Vert \cR  -\cG\bgamma\Vert^2 + \lambda_A
    \Vert\tilde\bgamma\Vert^2
  \end{align}
with no constraints, and the solution is $\hat\bgamma=(\cG^\tran\cG+\lambda_A I_{R-1})^{-1}\cG^\tran\cR$.
The fit is given by $\cG\hat\bgamma=\cG(\cG^\tran\cG+\lambda_A
I_{R-1})^{-1}\cG^\tran\cR=\tilde \cS_A\cR$, and $\tilde \cS_A$ is
clearly a symmetric operator.

To obtain the simplified expression for $\hat\bsa$, we write
\begin{align}
\cG\hat\gamma&=\cz_AQ(Q^\tran\cz_A^\tran\cz_A Q+\lambda_A
                 I_{R-1})^{-1}Q^\tran
                 \cz_A^\tran\cR\nonumber\\
 &=\cz_AQ(Q^\tran D Q+\lambda_A
                 I_{R-1})^{-1}Q^\tran
                 \cR^+\label{eq:tosimplify}\\
             &=\cz_A\hat\bsa,\nonumber
\end{align}
with $D=\diag(\nid)$.
We write  $H=Q(Q^\tran D Q+\lambda_A I_{R-1})^{-1}Q^\tran$
and $\tilde
Q=(D+\lambda_A I_R)^{\frac12}Q$, and let
\begin{align}
  \tilde H&= (D+\lambda_A I_R)^{\frac12} H (D+\lambda_A
              I_R)^{\frac12}
= \tilde Q(\tilde Q^\tran\tilde Q)^{-1}\tilde
    Q^\tran.\label{eq:Qproj}
\end{align}
  Now (\ref{eq:Qproj}) is a projection matrix in $\real^R$ onto a
    $R-1$ dimensional subspace. Let $\tilde q = (D+\lambda_A
    I_R)^{-\frac12}\bone.$ Then $\tilde q^\tran \tilde Q={\bszero}$, and so
    $$\tilde H=I_R-\frac{\tilde q\tilde q^\tran}{\Vert \tilde
      q\Vert^2}.$$
    Unraveling this expression we get
    $$ H=(D+\lambda_AI_R)^{-1}
    -(D+\lambda_AI_R)^{-1}\frac{\bone\bone^\tran}{\bone^\tran(D+\lambda_AI_R)^{-1}\bone}(D+\lambda_AI_R)^{-1}.$$
    With $\hat\bsa=H\cR^+$ in (\ref{eq:tosimplify}), this gives the
    expressions for each $\hat a_i$ in~\eqref{eq:ahatsoln}.
Finally, $\tilde \cS_A = \cz_A H\cz_A^\tran$ is symmetric.
\end{proof}

\subsection{Covariance matrix for $\hat\beta_{\gls}$ with centered operators}
\label{sec:covar-matr-hatb}
In Section~\ref{sec:two-factors} we saw in (\ref{eq:covbhat}) that we
get a simple expression for
$\cov(\hat\beta_{\gls})$. This simplicity relies on the fact that
$I_N-\cS_G=\cw=\sse\cv^{-1}$, and the usual cancelation occurs when we
use the sandwich formula to compute this covariance.
When we backfit with our centered smoothers we get a modified residual
operator
$I_N-\widetilde \cS_G$ such that the analog of (\ref{eq:gfactor})
still gives us the required coefficient estimate:
\begin{equation}
  \label{eq:gfactorc}
\hat\beta_{\gls} = (\cx^\tran(I_N-\widetilde\cS_G)\cx)^{-1}\cx^\tran(I_N-\widetilde\cS_G)\cy.
\end{equation}
However, $I_N-\widetilde\cS_G\neq \sse\cv^{-1}$, and so now we need to
resort to the sandwich formula
$ \cov(\hat\beta_{\gls})=H_\gls \cv H_\gls^\tran$
with $H_\gls$ from \eqref{eq:gfactor}.
Expanding this we find that
$\cov(\hat\beta_{\gls})$ equals
\begin{align*}
 (\cx^\tran(I_N-\widetilde\cS_G)\cx)^{-1}\cx^\tran(I_N-\widetilde\cS_G)
  \cdot \cv\cdot (I_N-\widetilde\cS_G)\cx(\cx^\tran(I_N-\widetilde\cS_G)\cx)^{-1}.
\end{align*}

While this expression might appear daunting, the computations are simple.
Note first that while $\hat\beta_{\gls}$ can be computed via
$\tilde\cS_G\cx$ and $\tilde\cS_G\cy$ this expression for $\cov(\hat\beta_{\gls})$
also involves $\cx^\tran \tilde\cS_G$.  When we use the centered operator
from Theorem~\ref{thm:smartcenter} we get a symmetric matrix $\tilde \cS_G$.
 Let $\widetilde \cx=(I_N-\widetilde\cS_G)\cx$, the residual
matrix after backfitting each column of $\cx$ using these centered operators. Then because
$\widetilde\cS_G$ is symmetric, we have
\begin{align}
  \hat\beta_{\gls}&=(\cx^\tran\widetilde\cx)^{-1}\widetilde\cx^\tran\cy,\quad\text{and}   \notag\\
   \cov(\hat\beta_{\gls})&=(\cx^\tran\widetilde\cx)^{-1}\widetilde\cx^\tran\cdot\cv\cdot\widetilde\cx(\cx^\tran\widetilde\cx)^{-1}.\label{eq:covbhatgls}
\end{align}
Since $\cv=\sse\left(\cz_A\cz_A^\tran/\lambda_A+\cz_B\cz_B^\tran/\lambda_B
  +I_N\right)$ (two low-rank matrices plus the identity), we can
compute $\cv\cdot \widetilde\cx$ very efficiently, and hence also the
covariance matrix in~\eqref{eq:covbhatgls}.
The entire algorithm is summarized in Section~\ref{sec:wholeshebang}.

\section{Convergence of the matrix norm}\label{sec:normconvergence}

In this section we prove a bound on the norm of the matrix
that implements backfitting for our random effects $\bsa$ and $\bsb$
and show how this controls the number of iterations required.
In our algorithm, backfitting is applied to $\cy$ as well as to each non-intercept column of $\cx$
so we do not need to consider the updates for $\cx\hat\beta$.
It is useful to take account of intercept adjustments in backfitting,
by the centerings described in Section~\ref{sec:backfitting}
because the space spanned by  $a_1,\dots,a_R$
intersects the space spanned by $b_1,\dots,b_C$ because both
include an intercept column of ones.

In backfitting we alternate between adjusting $\bsa$ given $\bsb$ and
$\bsb$ given $\bsa$. At any iteration, the new $\bsa$ is an affine function of
the previous $\bsb$
and then the new $\bsb$ is an affine function of the new $\bsa$.
This makes the new $\bsb$ an affine function of the previous $\bsb$.
We will study that affine function to find conditions where
the updates converge. If the $\bsb$ updates converge, then so must the $\bsa$
updates.

Because the updates are affine they can be written in the form
$$
\bsb \gets M\bsb + \eta
$$
for $M\in\real^{C\times C}$ and $\eta\in\real^C$.
We iterate this update and
it is convenient to start with $\bsb = \bszero$.
We already know from \cite{buja:hast:tibs:1989} that this  backfitting
will converge.  However, we want more. We want to avoid
having the number of iterations required grow with $N$.
We can write the solution $\bsb$ as
$$
\bsb = \eta +\sum_{k=1}^\infty M^k\eta,
$$
and in computations we truncate this sum after $K$ steps
producing an error $\sum_{k>K}M^k\eta$.  We want
$\sup_{\eta\ne0}\Vert \sum_{k>K}M^k\eta\Vert/\Vert\eta\Vert<\epsilon$
to hold with probability tending to one as the sample size
increases for any $\epsilon$, given sufficiently large $K$.
For this it suffices to have the spectral
radius $\lambda_{\max}(M)<1-\delta$ hold with probability
tending to one for some $\delta>0$.

Now for any $1\le p\le\infty$ we have
$$
\lambda_{\max}(M)\le \Vert M\Vert_{p}
\equiv \sup_{\bsx\in \real^C\setminus\{\bszero\}}
\frac{\Vert M\bsx\Vert_p}{\Vert \bsx\Vert_p}.
$$
The explicit formula
$$
\Vert M\Vert_{1}
\equiv \sup_{\bsx\in \real^C\setminus\{\bszero\}}
\frac{\Vert M\bsx\Vert_1}{\Vert \bsx\Vert_1}
= \max_{1\le s\le C}\sum_{j=1}^C | M_{js}|
$$
makes the matrix $L_1$  matrix norm very tractable theoretically
and so that is the one we study.  We look at this and some
other measures numerically in Section~\ref{sec:empiricalnorms}.




\subsection{Updates}
Recall that $Z\in\{0,1\}^{R\times C}$ describes the pattern of observations.
In a model with no intercept, centering the responses and
then taking shrunken means as in \eqref{eq:backfit} would yield
these updates
\begin{align*}
a_i &\gets \frac{\sum_s \zis(\yis-\bs)}{\nid+\lambda_A}\quad\text{and}\quad
b_j \gets \frac{\sum_i \zij(\yij-\ai)}{\ndj+\lambda_B}.
\end{align*}
The update from the old $\bsb$ to the new $\bsa$ and
then to the new $\bsb$
takes the form $\bsb\gets M\bsb+\eta$ for
$M=\mzero$ where
$$
\mzero_{js} =
\frac1{\ndj+\lambda_B}\sum_i \frac{\zis\zij}{\nid+\lambda_A}.$$
This update $\mzero$ alternates shrinkage estimates for $\bsa$
and $\bsb$ but does no centering.
We don't exhibit $\eta$ because it does not affect the
convergence speed.

In the presence of an intercept, we know that $\sum_ia_i=0$
should hold at the solution and we can impose this simply
and very directly by centering the $a_i$, taking
\begin{align*}
a_i &\gets \frac{\sum_s \zis(\yis-\bs)}{\nid+\lambda_A}
-\frac1R\sum_{r=1}^R\frac{\sum_s \zrs(\yrs-\bs)}{\nrd+\lambda_A},
\quad\text{and}\\
b_j &\gets \frac{\sum_i \zij(\yij-\ai)}{\ndj+\lambda_B}.
\end{align*}
The intercept estimate will then be $\hat\beta_0=(1/C)\sum_jb_j$ which
we can subtract from $b_j$ upon convergence.
This iteration has the update matrix $\mone$ with
\begin{align}\label{eq:monejs}
\mone_{js}
&=\frac1{\ndj+\lambda_B}\sum_r
\frac{\zrs(\zrj-\ndj/R)}{\nrd+\lambda_A}
\end{align}
after replacing a sum over $i$ by an equivalent one over $r$.

In practice, we prefer to use the weighted centering from
Section~\ref{sec:centered-operators} to center the $a_i$
because it provides a symmetric smoother $\tilde\cS_G$
that supports computation of $\wh\cov(\hat\beta_{\gls})$.
While it is more complicated to analyze it is easily computable
and it satisfies the optimality condition in Theorem~\ref{thm:smartcenter}.
The algorithm is for a generic response $\cR\in\real^N$ such as $\cy$
or a column of $\cx$.
Let us illustrate it for the case $\cR=\cy$.
We begin with vector of $N$ values $\yij-\bj$
and so $Y^+_i = \sum_s\zis(\yis-\bs).$
Then
$w_i = (\nid+\lambda_A)^{-1}/\sum_r(\nrd+\lambda_A)^{-1}$
and the updated $a_r$ is
\begin{align*}
\frac{Y^+_r-\sum_iw_i Y^+_i}{\nrd+\lambda_A}
&=
\frac{\sum_s\zrs(\yrs-\bs)-\sum_iw_i
\sum_s\zis(\yis-\bs)}{\nrd+\lambda_A}.
\end{align*}
Using shrunken averages of $\yij-\ai$, the new $\bj$ are
\begin{align*}
\bj &=\frac1{\ndj+\lambda_B}\sum_r\zrj
\biggl(\yrj-
\frac{\sum_s\zrs(\yrs-\bs)-\sum_iw_i
\sum_s\zis(\yis-\bs)}{\nrd+\lambda_A}
\biggr).
\end{align*}
Now $\bsb \gets M\bsb+\eta$ for $M=\mtwo$, where
\begin{align}\label{eq:mtwojs}
\mtwo_{js}
&=\frac1{\ndj+\lambda_B}\sum_r
\frac{\zrj}{\nrd+\lambda_A}
\biggl(\zrs - \frac{\sum_{i}\frac{\zis}{\nid+\lambda_{A}}}{\sum_i{\frac{1}{\nid+\lambda_{A}}}}\biggr).
\end{align}

Our preferred algorithm applies the optimal update
from Theorem~\ref{thm:smartcenter}
to both $\bsa$ and $\bsb$ updates.  With that choice we do
not need to decide beforehand which random effects to center
and which to leave uncentered to contain the intercept.
We call the corresponding matrix $\mthree$.
Our theory below analyzes $\Vert\mone\Vert_1$
and $\Vert\mtwo\Vert_1$
which have simpler expressions than
$\Vert\mthree\Vert_1$.

Update $\mzero$ uses symmetric smoothers
for $A$ and $B$. Both are shrunken
averages.  The naive centering update $\mone$ uses
a non-symmetric smoother
$\cz_A(I_R-\bone_R\bone_R^\tran/R)(\cz_A^\tran\cz_A+\lambda_AI_R)^{-1}\cz_A^\tran$
on the $\ai$ with a symmetric smoother on $\bj$
and hence it does not generally produce a symmetric
smoother needed for  efficient computation
of $\wh\cov(\hat\beta_{\gls})$.
The update $\mtwo$ uses two symmetric
smoothers, one optimal and one a simple shrunken mean.
The update $\mthree$ takes the optimal
smoother for both $A$ and $B$.
Thus both $\mtwo$ and $\mthree$
support efficient computation of $\wh\cov(\hat\beta_{\gls})$.
A subtle point is that these symmetric smoothers are
matrices in $\real^{N\times N}$ while the matrices $M^{(k)}\in\real^{C\times C}$
are not symmetric.

\subsection{Model for $\zij$}\label{sec:modelz}
We will state conditions on $\zij$ under which
both $\Vert \mone\Vert_1$ and $\Vert \mtwo\Vert_1$
are bounded
below $1$ with probability tending to one, as the problem size grows.
We need the following exponential inequalities.

\begin{lemma}\label{lem:hoeff}
If $X\sim\dbin(n,p)$, then for any $t\ge0$,
\begin{align*}
\Pr( X\ge np+t ) &\le \exp( -2t^2/n ),\quad\text{and}\\
\Pr( X\le np-t ) &\le \exp( -2t^2/n ).
\end{align*}
\end{lemma}
\begin{proof}
This follows from Hoeffding's theorem.
\end{proof}

\begin{lemma}\label{lem:binounionbound}
Let $X_i\sim\dbin(n,p)$ for $i=1,\dots,m$, not necessarily independent.
Then for any $t\ge0$,
\begin{align*}
\Pr\Bigl( \max_{1\le i\le m} X_{i} \ge np+t \Bigr) &\le m\exp( -2t^2/n ) ,\quad\text{and}\\
\Pr\Bigl( \min_{1\le i\le m} X_{i} \le np-t \Bigr) &\le m\exp( -2t^2/n ).
\end{align*}
\end{lemma}
\begin{proof}
  This is from the union bound applied
  to Lemma~\ref{lem:hoeff}.
\end{proof}

Here is our sampling model.
We index the size of our problem by $S\to\infty$.
The sample size $N$ will satisfy $\e(N)\ge S$.
The number of rows and columns in the data set are
$$R = S^\rho\quad\text{and}\quad C=S^\kappa$$
respectively, for positive numbers $\rho$ and $\kappa$.
Because our application domain has $N\ll RC$, we
assume that $\rho+\kappa>1$.
We ignore that $R$ and $C$ above are not necessarily integers.

In our model, $\zij\sim\dbern(\pij)$ independently with
\begin{align}\label{eq:defab}
\frac{S}{RC} \le \pij \le \pup\frac{S}{RC}
\quad\text{for}\quad 1\le\pup<\infty.
\end{align}
That is $1\le \pij S^{\rho+\kappa-1}\le\pup$.
Letting $\pij$ depend on $i$ and $j$
allows the probability model to capture
stylistic preferences affecting the missingness
pattern in the ratings data.

\subsection{Bounds for row and column size}
Letting $X \stocleq Y$ mean that $X$ is stochastically smaller than $Y$, we know that
\begin{align*}
\dbin(R, S^{1-\rho-\kappa}) &\stocleq \ndj \stocleq \dbin( R, \pup S^{1-\rho-\kappa}),\quad\text{and}\\
\dbin(C,S^{1-\rho-\kappa}) &\stocleq \nid \stocleq \dbin( C, \pup S^{1-\rho-\kappa}).
\end{align*}
By Lemma  \ref{lem:hoeff}, if $t\ge0$, then
\begin{align*}
\Pr( \nid \ge S^{1-\rho}(\pup+t))
&\le \Pr\bigl( \dbin(C,\pup S^{1-\rho-\kappa}) \ge S^{1-\rho}(\pup+t)\bigr)\\
&\le \exp(-2(S^{1-\rho}t)^2/C)\\
&= \exp(-2S^{2-\kappa-2\rho}t^2).
\end{align*}
Therefore if $2\rho+\kappa<2$, we find using
using Lemma~\ref{lem:binounionbound} that
\begin{align*}
 &\Pr\bigl( \max_i\nid \ge S^{1-\rho}(\pup+\epsilon)\bigr)
\le S^\rho\exp(-2S^{2-\kappa-2\rho}\epsilon^2)\to0
\end{align*}
for any $\epsilon>0$.
Combining this with an analogous lower bound,
\begin{align}\label{eq:boundnid}
\lim_{S\to\infty}\Pr\bigl( (1-\epsilon) S^{1-\rho}\le \min_i \nid \le \max_i \nid \le (\pup+\epsilon) S^{1-\rho}\bigr)=1.
\end{align}
Likewise, if $\rho+2\kappa<2$, then for any $\epsilon>0$,
\begin{align}\label{eq:boundndj}
\lim_{S\to\infty}\Pr\bigl( (1-\epsilon)S^{1-\kappa}\le \min_j \ndj \le \max_j \ndj \le (\pup+\epsilon) S^{1-\kappa}\bigr)=1.
\end{align}

\subsection{Interval arithmetic}
We will replace $\nid$ and other quantities
by intervals that asymptotically contain them with probability one and then use interval arithmetic in order to streamline
some of the steps in our proofs.
For instance,
$$\nid\in [(1-\epsilon)S^{1-\rho},(\pup+\epsilon)S^{1-\rho}]
= [1-\epsilon,\pup+\epsilon]\times S^{1-\rho}
= [1-\epsilon,\pup+\epsilon]\times \frac{S}{R}$$
holds simultaneously for all $1\le i\le R$ with probability
tending to one as $S\to\infty$.
In interval arithmetic,
$$[A,B]+[a,b]=[a+A,b+B]\quad\text{and}\quad [A,B]-[a,b]=[A-b,B-a].$$
If $0<a\le b<\infty$ and $0<A\le B<\infty$, then
$$[A,B]\times[a,b] = [Aa,Bb]\quad\text{and}\quad [A,B]/[a,b] = [A/b,B/a].$$
Similarly, if $a<0<b$ and $X\in[a,b]$, then
$|X|\in[0,\max(|a|,|b|)]$.
Our arithmetic operations on intervals yield
new intervals guaranteed to contain the results
obtained using any members of the original intervals.
We do not necessarily use the smallest such interval.

\subsection{Co-observation}

Recall that the co-observation matrices are $Z^\tran Z\in\{0,1\}^{C\times C}$
and $ZZ^\tran\in\{0,1\}^{R\times R}$.
If $s\ne j$, then
$$
\dbin\Bigl( R,\frac{S^2}{R^2C^2}\Bigr)
\stocleq \ztz_{sj}\stocleq
\dbin\Bigl( R,\frac{\pup^2S^2}{R^2C^2}\Bigr).
$$
That is
$\dbin(S^\rho, S^{2-2\rho-2\kappa})
\stocleq
\ztz_{sj}
\stocleq
\dbin(S^\rho, \pup^2S^{2-2\rho-2\kappa}).
$
For $t\ge0$,
\begin{align*}
\Pr\Bigl( \max_s\max_{j\ne s}\ztz_{sj}\ge (\pup^2+t)S^{2-\rho-2\kappa}\Bigr)
&\le \frac{C^2}2\exp( -(tS^{2-\rho-2\kappa})^2/R)\\
&= \frac{C^2}2\exp( -t^2 S^{4-3\rho-4\kappa}).
\end{align*}
If $3\rho+4\kappa<4$
then
\begin{align*}
&\Pr\Bigl( \max_s\max_{j\ne s} \,\ztz_{sj} \ge (\pup^2+\epsilon)S^{2-\rho-2\kappa}\Bigr)\to0,
\quad\text{and}\\
&\Pr\Bigl( \min_s\min_{j\ne s} \,\ztz_{sj} \le (1-\epsilon)S^{2-\rho-2\kappa}\Bigr)\to0,
\end{align*}
for any $\epsilon>0$.

\subsection{Asymptotic bounds for $\Vert M\Vert_1$}
Here we prove upper bounds for $\Vert M^{(k)}\Vert_1$ for $k=1,2$
of equations~\eqref{eq:monejs} and~\eqref{eq:mtwojs}, respectively.
The bounds depend on $\pup$ and there are values of $\pup>1$
for which  these norms are bounded strictly below one,
with probability tending to one.

\begin{theorem}\label{thm:m1norm1}
Let $\zij$ follow the model from Section~\ref{sec:modelz}
with $\rho,\kappa\in(0,1)$, that satisfy $\rho+\kappa>1$,  
$2\rho+\kappa<2$ and $3\rho+4\kappa<4$.
Then for any $\epsilon>0$,
\begin{align}\label{eq:claim1}
&
\Pr\bigl( \Vert \mone \Vert_1\le
\pup^2-\pup^{-2}+\epsilon
\bigr)\to1
,\quad\text{and}\\
&\Pr\bigl( \Vert \mtwo\Vert_1\le
\pup^2-\pup^{-2}+\epsilon\bigr)\to1 \label{eq:claim2}
\end{align}
as $S\to\infty$.
\end{theorem}

\begin{figure}[t!]
  \centering
   \includegraphics[width=.8\hsize]{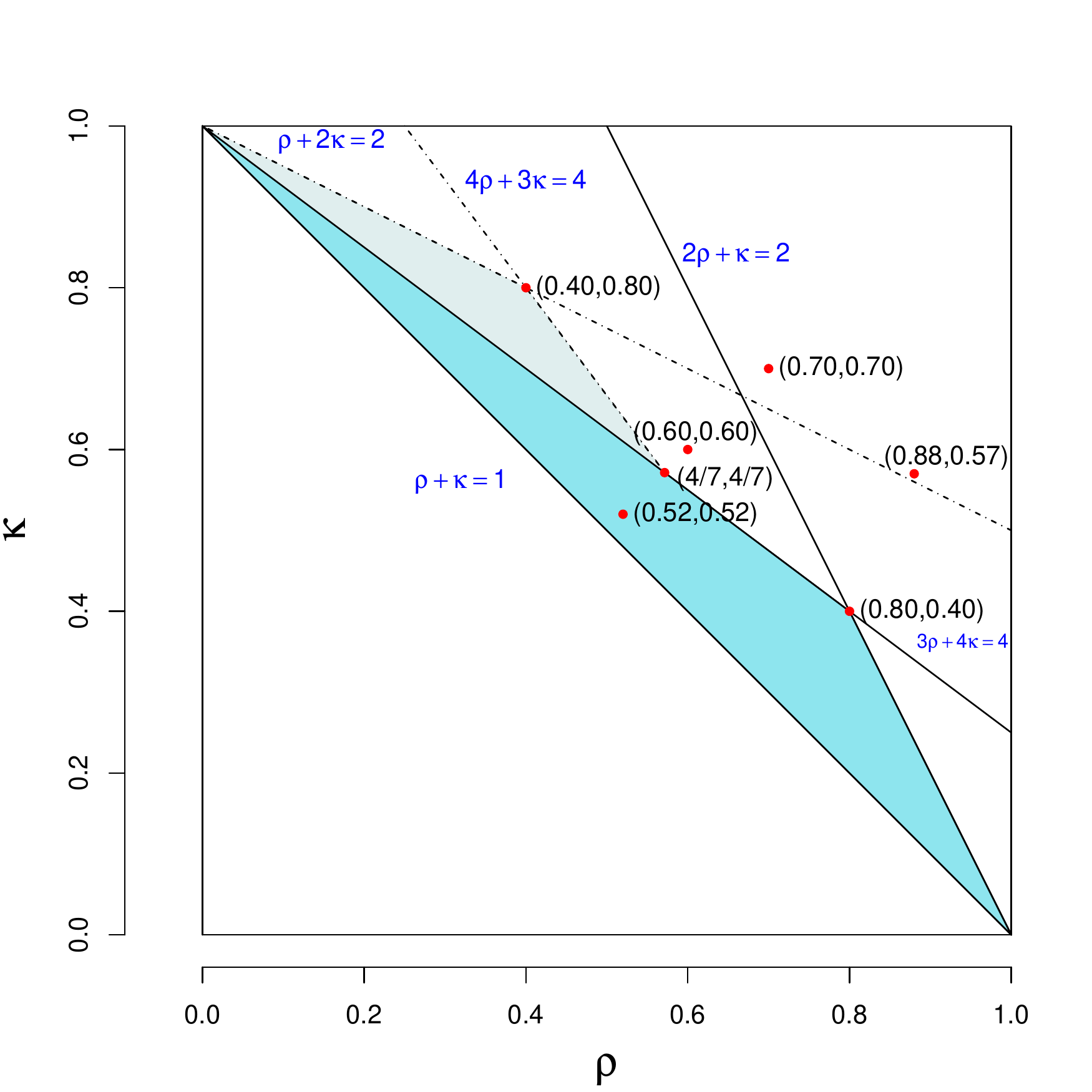}
    \caption{
\label{fig:domainofinterest}
The large shaded triangle is the domain of interest $\dom$ for
Theorem~\ref{thm:m1norm1}.
The smaller shaded triangle shows a region where the analogous update
to $\bsa$ would have acceptable norm.  The points marked are the ones we look at numerically,
including $(0.88,0.57)$ which corresponds to the Stitch Fix data in
Section~\ref{sec:stitch}.
}
\end{figure}

\begin{proof}
Without loss of generality we assume that $\epsilon<1$.
We begin with~\eqref{eq:claim2}.
Let $M=\mtwo$.
When $j\ne s$,
\begin{align*}
M_{js}&=\frac1{\ndj+\lambda_B}\sum_r
\frac{\zrj}{\nrd+\lambda_A}
(\zrs -\bar Z_{\sumdot s}),\quad\text{for}\\
\bar Z_{\sumdot s}&=
\sum_i
\frac{\zis}{\nid+\lambda_A}
\Bigm/
{\sum_{i}\frac{1}{\nid+\lambda_{A}}}.
\end{align*}
Although $|\zrs-\bar Z_{\sumdot s}|\le1$, replacing
$\zrs-\bar Z_{\sumdot s}$ by one does not prove to be
sharp enough for our purposes.

Every $\nrd+\lambda_A\in S^{1-\rho} [1-\epsilon, \pup+\epsilon]$
with probability tending to one and so
\begin{align*}
\frac{\bar Z_{\sumdot s}}{\ndj+\lambda_B}\sum_r
\frac{\zrj}{\nrd+\lambda_A}
&\in
\frac{\bar Z_{\sumdot s}}{\ndj+\lambda_B}\sum_r
\frac{\zrj}{[1-\epsilon,\pup+\epsilon]S^{1-\rho}}\\
&\subseteq [1-\epsilon,\pup+\epsilon]^{-1}\bar Z_{\sumdot s} S^{\rho-1}.
\end{align*}
Similarly
\begin{align*}
\bar Z_{\sumdot s} &\in
\frac{\sum_i\zis[1-\epsilon,\pup+\epsilon]^{-1}}
{R[1-\epsilon,\pup+\epsilon]^{-1}}
\subseteq\frac{\nds}{R}[1-\epsilon,\pup+\epsilon][1-\epsilon,\pup+\epsilon]^{-1}\\
&\subseteq S^{1-\rho-\kappa}
[1-\epsilon,\pup+\epsilon]^2[1-\epsilon,\pup+\epsilon]^{-1}
\end{align*}
and so
\begin{align}\label{eq:zrsbarpart}
\frac{\bar Z_{\sumdot s}}{\ndj+\lambda_B}\sum_r
\frac{\zrj}{\nrd+\lambda_A}
\in S^{-\kappa}
\frac{[1-\epsilon,\pup+\epsilon]^2}{[1-\epsilon,\pup+\epsilon]^2}
\subseteq \frac1C
\Bigl[
\Bigl(\frac{1-\epsilon}{\pup+\epsilon}\Bigr)^2
, \Bigl(\frac{\pup+\epsilon}{1-\epsilon}\Bigr)^2
\Bigr].
\end{align}
Next using bounds on the co-observation counts,
\begin{align}\label{eq:zrspart}
\frac1{\ndj+\lambda_B}\sum_r\frac{\zrj\zrs}{\nrd+\lambda_A}
\in \frac{S^{\rho+\kappa-2}\ztz_{sj}}{[1-\epsilon,\pup+\epsilon]^2}
\subseteq
\frac1C
\frac{[1-\epsilon,\pup^2+\epsilon]}{[1-\epsilon,\pup+\epsilon]^2}.
\end{align}
Combining~\eqref{eq:zrsbarpart} and~\eqref{eq:zrspart}
\begin{align*}
M_{js} \in &
\frac1C
\Bigl[
\frac{1-\epsilon}{(\pup+\epsilon)^2}-
\Bigl(\frac{\pup+\epsilon}{1-\epsilon}\Bigr)^2
,
\frac{\pup^2+\epsilon}{1-\epsilon}
-\Bigl(\frac{1-\epsilon}{\pup+\epsilon}\Bigr)^2
\Bigr]
\end{align*}
For any $\epsilon'>0$ we can choose $\epsilon$ small enough that
$$M_{js} \in C^{-1}[\pup^{-2}-\pup^2-\epsilon',
\pup^2-\pup^{-2}+{\epsilon'}]
$$
and then $|M_{js}|\le (\pup^2-\pup^{-2}+\epsilon')/C$.

Next, arguments like the preceding
give  $|M_{jj}|\le (1-\epsilon')^{-2}(\pup+\epsilon')S^{\rho-1}\to0$.
Then with probability tending to one,
$$
\sum_j|M_{js}|
\le\pup^2-\pup^{-2}
+2\epsilon'.
$$
This bound holds for all $s\in\{1,2,\dots,C\}$, establishing~\eqref{eq:claim2}.

The proof of~\eqref{eq:claim1} is similar.
The quantity $\bar Z_{\sumdot s}$
is replaced by $(1/R)\sum_i\zis/(\nid+\lambda_A)$.
\end{proof}

It is interesting to find the largest $\pup$ with
$\pup^2-\pup^{-2}\le1$. 
It is
$((1+5^{1/2})/2)^{1/2}\doteq 1.27$.

\section{Convergence and computation}\label{sec:empiricalnorms}

In this section we make some computations on synthetic data
following the probability model from Section~\ref{sec:normconvergence}.
First we study the norms of our update matrix $\mtwo$
which affects the number of iterations to convergence.
In addition to $\Vert\cdot\Vert_1$ covered in Theorem~\ref{thm:m1norm1}
we also consider $\Vert\cdot\Vert_2$, $\Vert\cdot\Vert_\infty$  and $\lambda_{\max}(\cdot)$.
Then we compare the cost to compute $\hat\beta_\gls$ by
our backfitting method with that of lmer \citep{lme4}.

The problem size is indexed by $S$.
Indices $i$ go from $1$ to $R=\lceil S^\rho\rceil$
and indices $j$ go from $1$ to $C=\lceil S^\kappa\rceil$.
Reasonable parameter values have $\rho,\kappa\in(0,1)$
with $\rho+\kappa>1$.
Theorem~\ref{thm:m1norm1} applies when
$2\rho+\kappa<2$ and $3\rho+4\kappa<4$.
Figure~\ref{fig:domainofinterest} depicts this
triangular domain of interest $\dom$.
There is another triangle $\dom'$ where a corresponding update
for $\bsa$ would satisfy the conditions of Theorem~\ref{thm:m1norm1}.
Then $\dom\cup\dom'$ is a non-convex polygon of five sides.
Figure~\ref{fig:domainofinterest}
also shows $\dom'\setminus\dom$ as a second triangular region.
For points $(\rho,\kappa)$ near the line $\rho+\kappa=1$, the matrix $Z$
will be mostly ones unless $S$ is very large.  For points $(\rho,\kappa)$
near the upper corner $(1,1)$, the matrix $Z$ will be extremely sparse
with each $\nid$ and $\ndj$ having nearly a
Poisson distribution with mean between $1$ and $\pup$.
The fraction of potential values that have been observed
is $O(S^{1-\rho-\kappa})$.

Given {$\pij$}, we generate our observation matrix via $\zij \simind\dbern({\pij)}$.
These probabilities are first generated via
${\pij}= U_{ij}S^{1-\rho-\kappa}$ where
$U_{ij}\simiid\dunif[1,\pup]$ and $\pup$ is
the largest value for which $\pup^2-\pup^{-2}\le1$.
For small $S$ and $\rho+\kappa$ near $1$ we can get
some values ${\pij>1}$ and in that case we take ${\pij=1}$.

The following $(\rho,\kappa)$ combinations are of interest.
First, $(4/5,2/5)$ is the closest vertex of the domain of interest to the point $(1,1)$.
Second, $(2/5,4/5)$ is outside the domain of interest for the $\bsb$
but within the domain for the analogous $\bsa$ update.
Third,  among points with $\rho=\kappa$, the value $(4/7,4/7)$
is the farthest one from the origin that is in the domain of interest.
We also look at some points on the $45$ degree line that are outside
the domain of interest because the sufficient conditions in
Theorem~\ref{thm:m1norm1}
might not be necessary.

In our matrix norm computations we took $\lambda_A=\lambda_B=0$.
This completely removes shrinkage and will make it harder for the algorithm to converge
than would be the case for the positive $\lambda_A$ and $\lambda_B$ that hold
in real data.  The values of $\lambda_A$ and $\lambda_B$
appear in expressions $\nid+\lambda_A$ and $\ndj+\lambda_B$ where their
contribution is asymptotically negligible, so conservatively setting them to zero
will nonetheless be realistic for large data sets.

\begin{figure}
  \centering
  \includegraphics[width=.8\hsize]{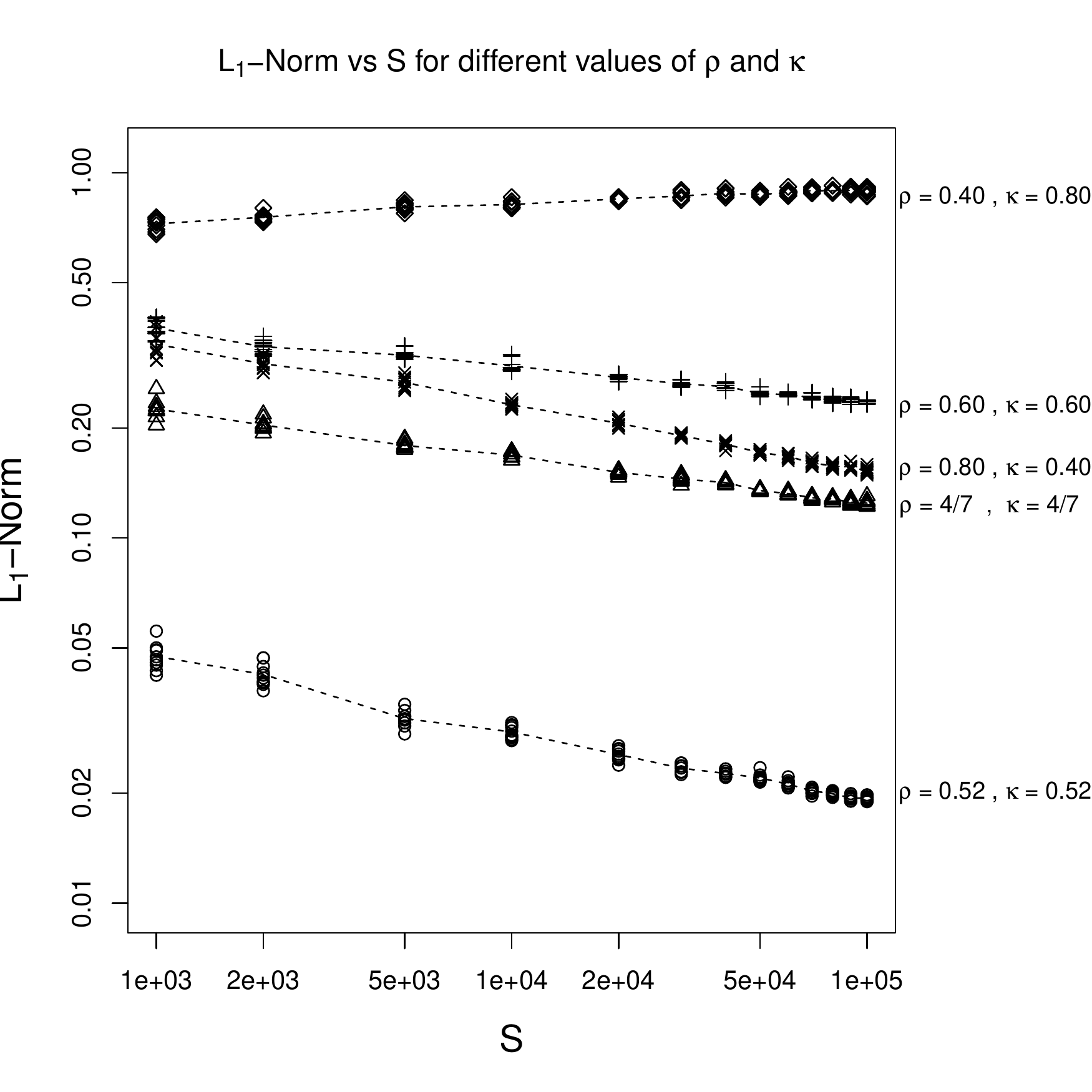}
\caption{\label{fig:1normvsn}
Norm
$\Vert M^{(2)}\Vert_1$ of centered update matrix
versus problem size $S$ for different $(\rho, \kappa)$.
}
\end{figure}
\noindent

We sample from the model multiple times at various values of $S$
and plot $\Vert \mtwo\Vert_1$ versus $S$ on a logarithmic scale.
Figure~\ref{fig:1normvsn} shows the results.
We observe that $\Vert \mtwo\Vert_1$ is below $1$ and decreasing
with $S$ for all the examples $(\rho,\kappa)\in\dom$.
This holds also for $(\rho,\kappa)=(0.60,0.60)\not\in\dom$.
We chose that point because it is on the convex hull of $\dom\cup\dom'$.

The point $(\rho,\kappa)=(0.40,0.80)\not\in\dom$.
Figure~\ref{fig:1normvsn} shows large values of $\Vert\mtwo\Vert_1$ for this
case. Those values increase with $S$, but remain below $1$ in the range considered.
This is a case where the update from $\bsa$ to $\bsa$ would have norm well below $1$
and decreasing with $S$, so backfitting would converge.
We do not know whether $\Vert\mtwo\Vert_1>1$ will occur for larger $S$.

The point $(\rho,\kappa)=(0.70,0.70)$ is not in the domain $\dom$
covered by Theorem~\ref{thm:m1norm1}
and we see that $\Vert\mtwo\Vert_1>1$ and generally increasing with $S$
as shown in Figure~\ref{fig:7070norms}.
This does not mean that backfitting must fail to converge.
Here we find  that $\Vert\mtwo\Vert_2<1$ and generally decreases as $S$
increases.  This is a strong indication that
the number of backfitting iterations required
will not grow with $S$ for this $(\rho,\kappa)$ combination.
We cannot tell whether $\Vert\mtwo\Vert_2$ will decrease to zero
but that is what appears to happen.

We consistently find in our computations
that $\lambda_{\max}(\mtwo)\le \Vert\mtwo\Vert_2\le\Vert\mtwo\Vert_1$.
The first of these inequalities must necessarily hold.
For a symmetric matrix $M$ we know that $\lambda_{\max}(M)=\Vert M\Vert_2$
which is then necessarily no larger than $\Vert M\Vert_1$.
Our update matrices are nearly symmetric but not perfectly so.
We believe that explains why their $L_2$ norms are
close to their spectral radius and also smaller than their $L_1$ norms.
While the $L_2$ norms are empirically more favorable than the $L_1$
norms, they are not amenable to our theoretical treatment.

\begin{figure}
  \centering
  \begin{subfigure}{.48\textwidth}
  \centering
  \includegraphics[scale=.4]{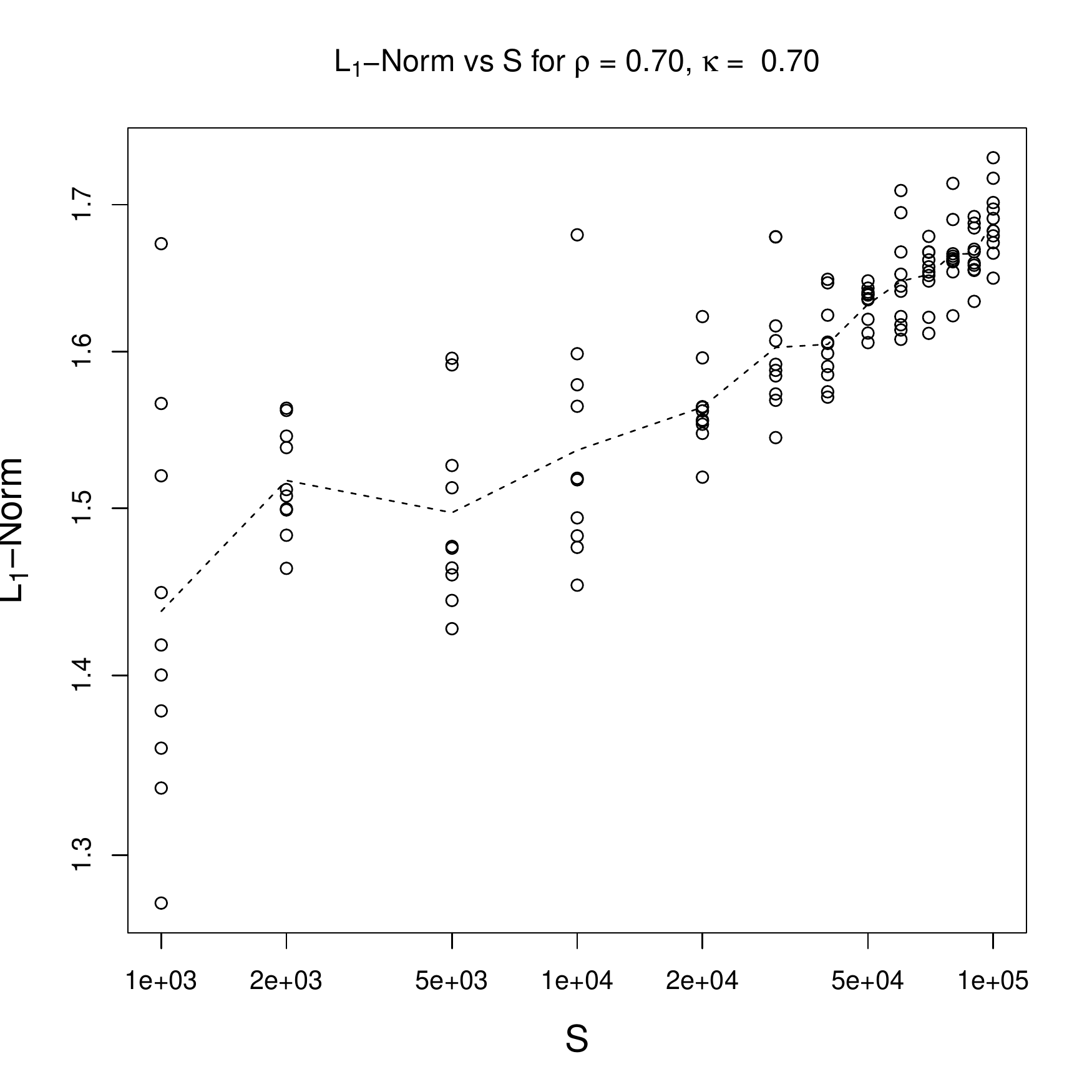}
 \end{subfigure}
\begin{subfigure}{.48\textwidth}
  \centering
\includegraphics[scale=.4]{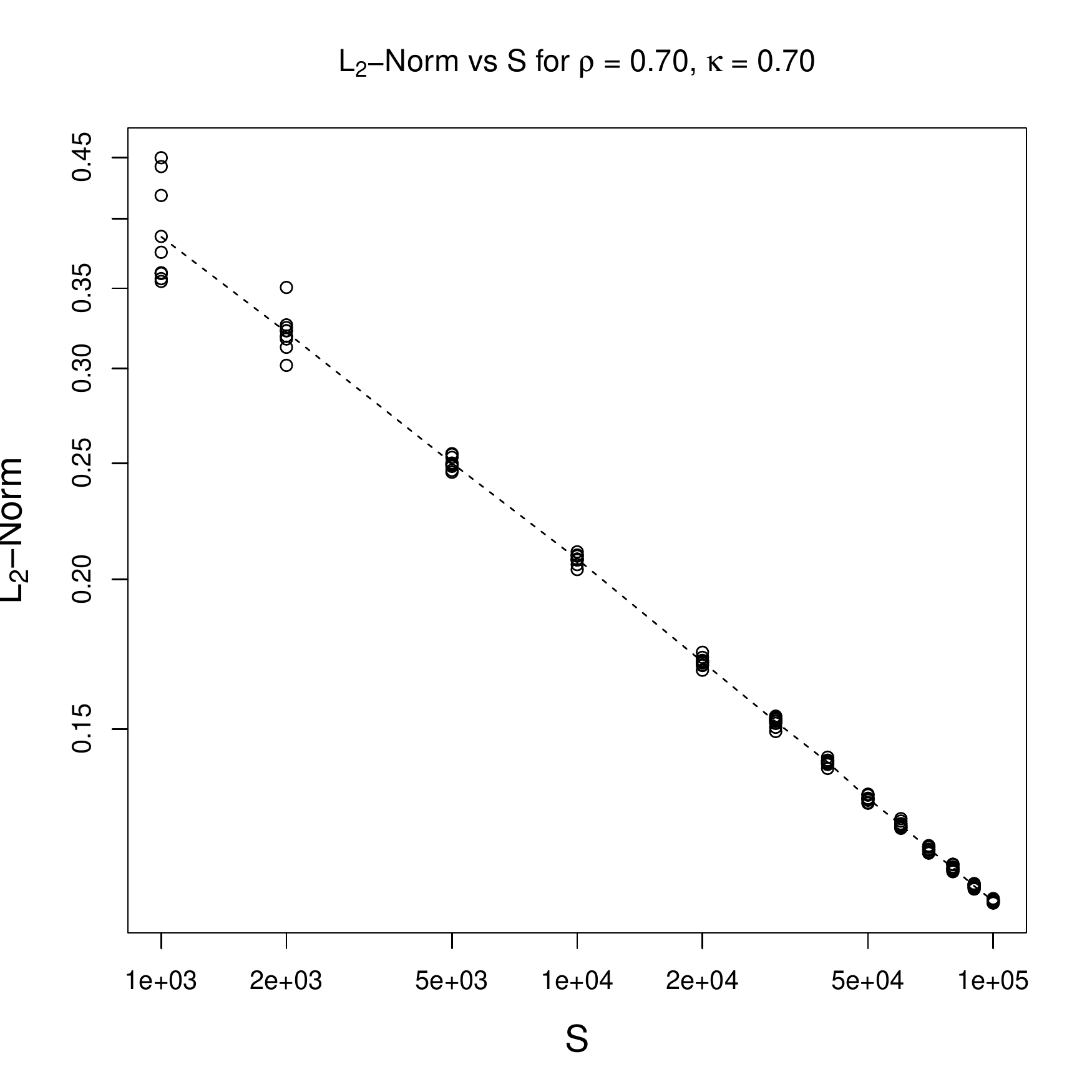}
\end{subfigure}
\caption{\label{fig:7070norms}
The left panel shows $\Vert\mtwo\Vert_1$ versus $S$.
The right panel shows $\Vert\mtwo\Vert_2$ versus $S$
with a logarithmic vertical scale.
Both have $(\rho,\kappa)=(0.7,0.7)$.
}
\end{figure}

We believe that backfitting will have a spectral radius well below $1$
for more cases than we can as yet prove.
In addition to the previous figures showing matrix norms
as $S$ increases for certain special values of $(\rho,\kappa)$ we
have computed contour maps of those norms over
$(\rho,\kappa)\in[0,1]$ for $S=10{,}000$.
See Figure~\ref{fig:contours}.

To compare the computation times for algorithms we
generated $\zij$ as above and also took
$\xij\simiid\dnorm(0,I_7)$ plus an intercept, making $p=8$
fixed effect parameters.
Although backfitting can run with $\lambda_A=\lambda_B=0$,
lmer cannot do so for numerical reasons.  So we took $\ssa=\ssb=1$
and $\sse=1$ corresponding to $\lambda_A=\lambda_B=1$.
The cost per iteration does not depend on $\yij$ and hence not
on $\beta$ either. We used $\beta=0$.

Figure~\ref{fig:comptimes} shows computation times
for one single iteration when $(\rho,\kappa)=(0.52,0.52)$ and when $(\rho,\kappa)=(0.70,0.70)$.
The time to do one iteration in lmer grows roughly like $N^{3/2}$
in the first case. For the second case, it appears to grow at
the  even faster rate of $N^{2.1}$.
Solving a system of $S^\kappa\times S^\kappa$ equations would cost
$S^{3\kappa} = S^{2.1} = O(N^{2.1})$, which explains the observed rate.
This analysis would predict $O(N^{1.56})$ for $\rho=\kappa=0.52$
but that is only minimally different from $O(N^{3/2})$.
These experiments were carried out in R on a computer
with the macOS operating system,  16 GB of memory and an Intel i7 processor. Each backfitting iteration entails solving \eqref{eq:backfit} along with the fixed effects.

The cost per iteration for backfitting follows closely to the  $O(N)$
rate predicted by the theory.
OLS only takes one iteration and it is also of
$O(N)$ cost.  In both of these cases $\Vert\mtwo\Vert_2$ is bounded away
from one so the number of backfitting iterations does not grow with $S$.
For $\rho=\kappa=0.52$,
backfitting took $4$ iterations to converge for the smaller values of $S$
and $3$ iterations for the larger ones.
For $\rho=\kappa=0.70$,
backfitting took $6$ iterations for smaller $S$ and $4$ or $5$ iterations
for larger $S$.
In each case our convergence criterion was a relative
change of $10^{-8}$
as described in Section~\ref{sec:wholeshebang}.
Further backfitting to compute BLUPs $\hat\bsa$ and $\hat\bsb$
given $\hat\beta_{\gls}$
took at most $5$ iterations for $\rho=\kappa=0.52$
and at most $10$ iterations for $\rho=\kappa=0.7$.
In the second example, lme4 did not reach convergence in
our time window so we ran it for just $4$ iterations to measure its cost per iteration.

\begin{figure}[!t]
  \centering
  \begin{subfigure}{.48\textwidth}
  \centering
  \includegraphics[scale=.28]{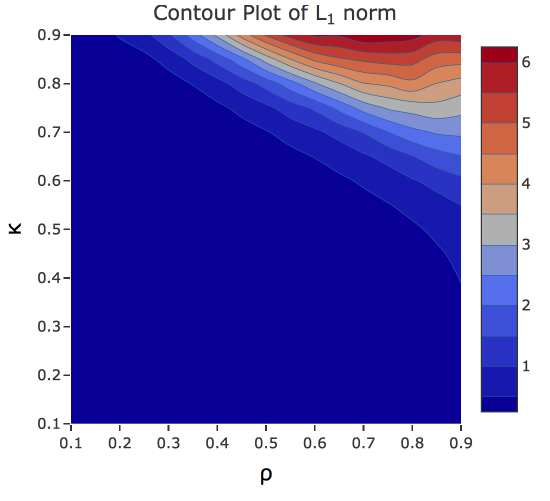}
 \end{subfigure}
\begin{subfigure}{.48\textwidth}
  \centering
\includegraphics[scale=.28]{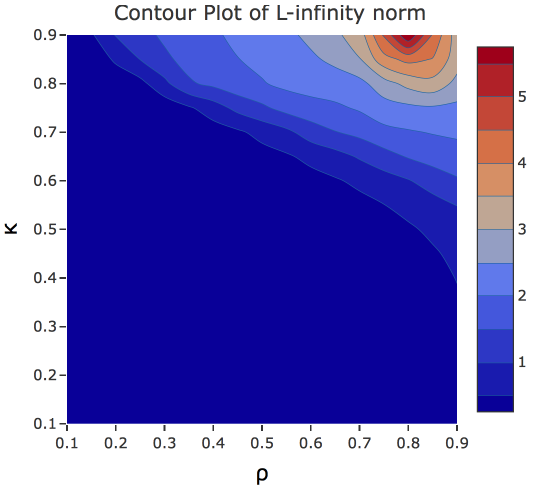}
\end{subfigure}
  \centering
  \begin{subfigure}{.48\textwidth}
  \centering
  \includegraphics[height = 5.2cm, width = 5.5cm]{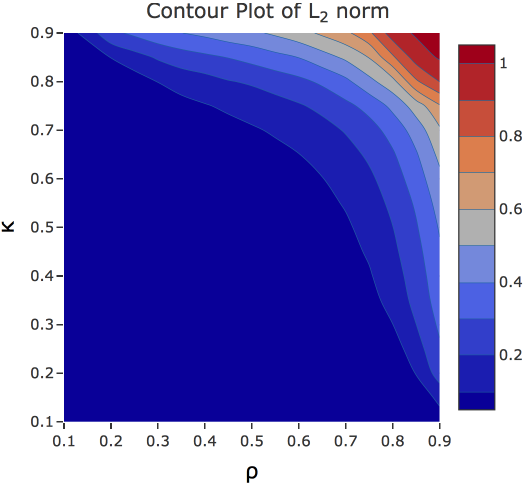}
 \end{subfigure}
\begin{subfigure}{.48\textwidth}
  \centering
\includegraphics[height = 5.2cm, width = 5.44cm]{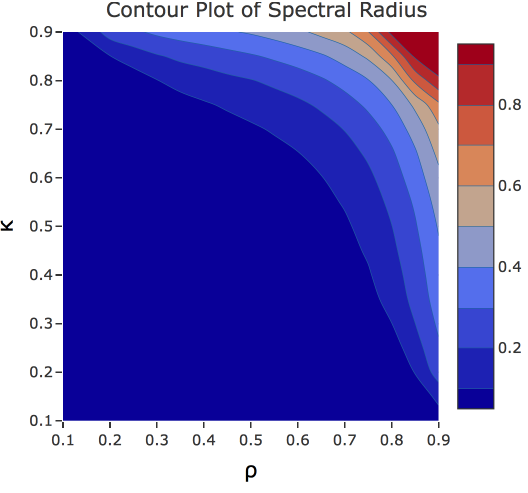}
\end{subfigure}
\caption{\label{fig:contours}
Numerically computed matrix norms
for $\mtwo$ using $S=10{,}000$.
The color code varies with the subfigures.
}
\end{figure}

\begin{figure}
\centering
   \begin{subfigure}{.48\textwidth}
  \centering
  \includegraphics[width=1\linewidth]{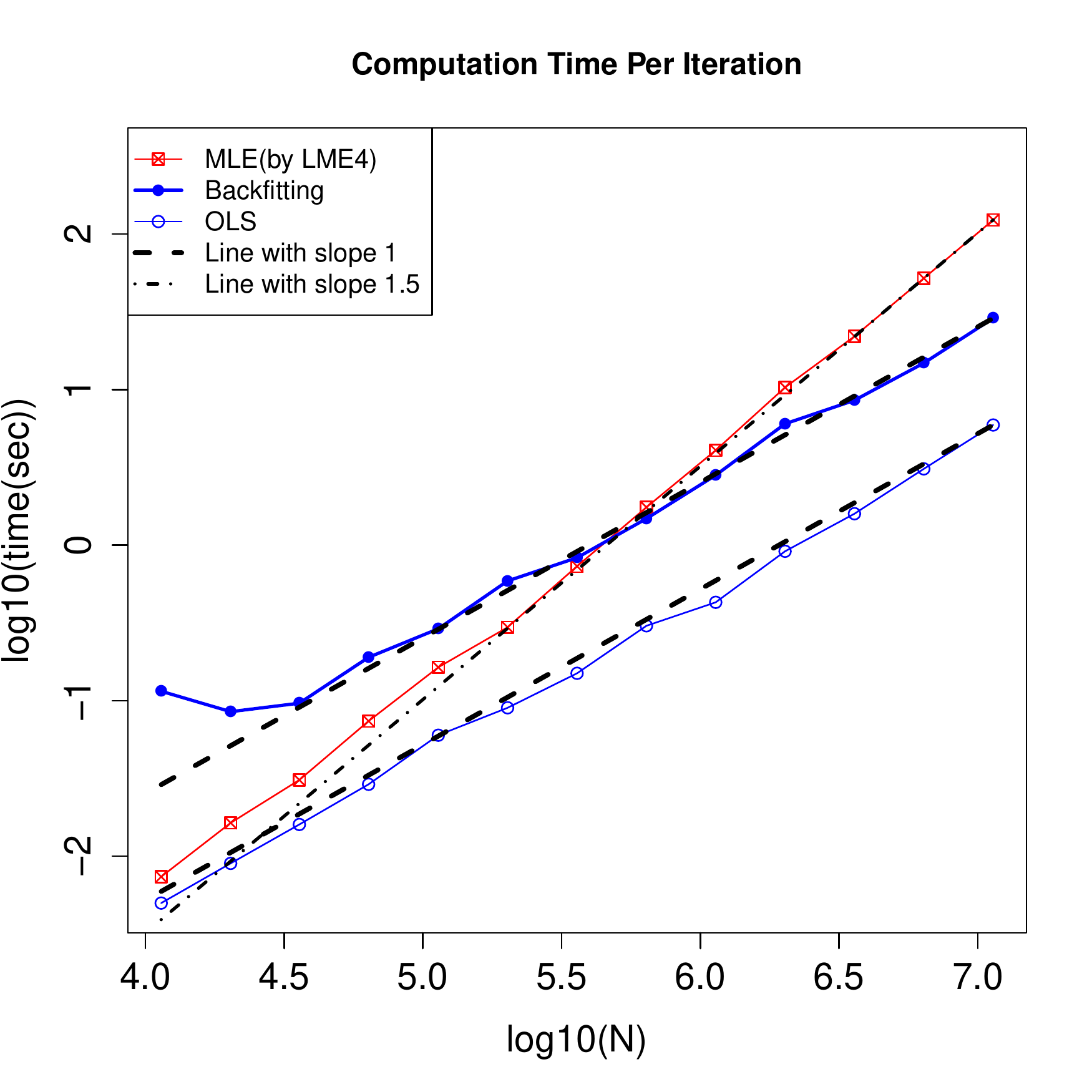}
      \caption{$(\rho, \kappa) = (0.52,0.52)$}
   \end{subfigure}
   
   \begin{subfigure}{.48\textwidth}
  \centering
\includegraphics[width=1\linewidth]{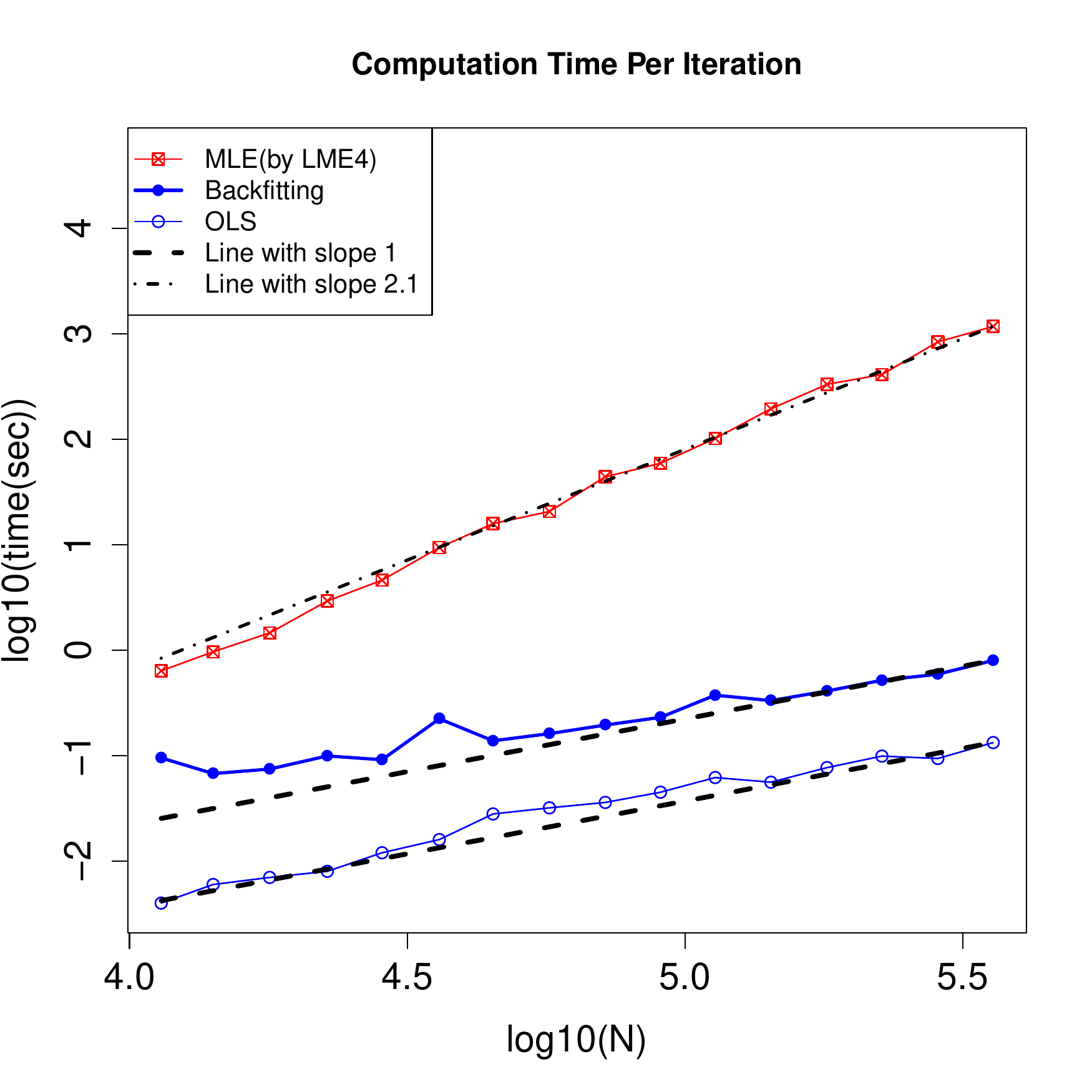}
      \caption{$(\rho, \kappa) = (0.70,0.70)$}
\end{subfigure}
\caption{\label{fig:comptimes}
Time for one iteration versus the number of observations, $N$ at two points $(\rho,\kappa)$.
The cost for lmer is roughly $O(N^{3/2})$ in the top panel
and $O(N^{2.1})$ in the bottom panel.  The costs for OLS and backfitting
are $O(N)$.
}
\end{figure}

\section{Example: ratings from Stitch Fix}\label{sec:stitch}

We illustrate backfitting for GLS on some data from Stitch Fix.
Stitch Fix sells  clothing. They mail their customers a sample of items.
The customers may keep and purchase any of those items that they
want, while returning the others.  It is valuable to predict
the extent to which a customer will like an item, not just whether they will purchase it.
Stitch Fix has provided us with some of their client ratings
data.  It was anonymized, void of personally identifying
information, and as a sample it does not reflect their
total numbers of clients or items at the time they
provided it.  It is also from 2015. While
it does not describe their current business, it is a valuable
data set for illustrative purposes.

The sample sizes for this  data are as follows.
We received $N=5{,}000{,}000$ ratings
by $R=762{,}752$ customers on $C=6{,}318$ items.
These values of $R$ and $C$ correspond to the point $(0.88,0.57)$  in Figure~\ref{fig:domainofinterest}.
Thus $C/N\doteq 0.00126$ and $R/N\doteq 0.153$.
The data are not dominated by a single row or column because
$\max_i\nid/R\doteq 9\times 10^{-6}$ and $\max_j\ndj/N\doteq 0.0143$.
The data are sparse because $N/(RC)\doteq 0.001$.

\subsection{An illustrative linear model}
The response $\yij$ is a rating on a ten point scale of
the satisfaction of customer $i$ with item $j$.
The data come with features about the clients and
items.  In a business setting one would fit and compare
possibly dozens of different regression models to understand the data.
Our purpose here is to study large scale GLS and compare
it to ordinary least squares (OLS) and so we use just one model, not necessarily
one that we would have settled on.
For that purpose we use the same model that was
used in \cite{crelin}. It is not chosen to make OLS look as bad as
possible.  Instead it is potentially the first model one might look at in
a data analysis.
For client $i$ and item $j$,
\begin{align}
\yij& =  \beta_0+\beta_1\mathrm{match}_{ij}+\beta_2\mathbb{I}\{\mathrm{client\  edgy}\}_i+\beta_3\mathbb{I}\{\mathrm{item\ edgy}\}_j \notag \\
&\phe +  \beta_4\mathbb{I}\{\mathrm{client\ edgy}\}_i*\mathbb{I}\{\mathrm{item\ edgy}\}_j+\beta_5\mathbb{I}\{\mathrm{client\ boho}\}_i \notag \\
&\phe + \beta_6\mathbb{I}\{\mathrm{item\ boho}\}_j+\beta_7\mathbb{I}\{\mathrm{client\  boho}\}_i*\mathbb{I}\{\mathrm{item\ boho}\}_j \notag \\
&\phe + \beta_8\mathrm{material}_{ij}+a_i+b_j+e_{ij}. \notag
\end{align}
Here $\mathrm{material}_{ij}$ is a categorical variable that is implemented via indicator variables for each type of material other than the baseline. Following \cite{crelin}, we chose ‘Polyester’, the most common material, as the baseline.
Some customers and some items were given the adjective `edgy' in the data set. Another adjective was `boho', short for `Bohemian'.
The variable match$_{ij}\in[0,1]$ is an estimate of the probability that the customer keeps the item, made before the item was sent.
The match score is a prediction from a baseline model and is not representative of all algorithms used at Stitch Fix.
All told, the model has $p=30$ parameters.

\subsection{Estimating the variance parameters}\label{sec:estim-vari-param}
We use the method of moments method from \cite{crelin}
to estimate $\theta^\tran=(\ssa, \ssb, \sse)$ in $O(N)$ computation.
That is in turn based on the method that
\cite{GO17} use in the intercept only model where
$\yij = \mu+\ai+\bj+\eij$.
For that model they set
\begin{align*}
U_{A} &= \sum_{i} \sum_{j} Z_{ij}
          \Bigl( Y_{ij}-\frac{1}{\nid}\sum_{j^{\prime}}\zijp
          Y_{ij^{\prime}}\Bigr)^{2}, \\
  U_{B} &= \sum_{j}\sum_{i}  Z_{ij}
            \Bigl(Y_{ij}-\frac{1}{\ndj}\sum_{i^{\prime}}\zipj
            Y_{i^{\prime}j}\Bigr)^{2}, \quad\text{and}\\
  U_{E} &= N\sum_{i j} Z_{i j}  \Bigl(Y_{i j}-\frac{1}{N}\sum_{i^{\prime} j^{\prime}}\zipjp Y_{i^{\prime} j^{\prime}}\Bigr)^{2}.
\end{align*}
These are, respectively, sums of within row sums of squares,
sums of within column sums of squares
and a scaled overall sum of squares.
Straightforward calculations
 show that
 \begin{align*}
   \mathbb{E}(U_{A})&=\bigl(\ssb+\sse\bigr)(N-R), \\
   \mathbb{E}(U_{B})&=\bigl(\ssa+\sse \bigr)(N-C), \quad\text{and}\\
   \mathbb{E}(U_{E})&=\ssa\Bigl(N^{2}-\sum_{i} \nid^{2}\Bigr)+\ssb\Bigl(N^{2}-\sum_{j} \ndj^{2}\Bigr)+\sse(N^{2}-N).
  \end{align*}
By matching moments, we can  estimate  $\theta$ by solving the $3 \times 3$ linear system
$$\begin{pmatrix}
0& N-R & N-R \\[.25ex]
N-C & 0 & N-C \\[.25ex]
N^{2}-\Sigma N_{i}^{2} & N^{2}-\Sigma N_{j}^{2} & N^{2}-N
\end{pmatrix}
\begin{pmatrix}
\ssa \\[.25ex] \ssb \\[.25ex] \sse\end{pmatrix}
=\begin{pmatrix}
U_{A}\\[.25ex] U_{B} \\[.25ex]  U_{E}\end{pmatrix}
$$
for $\theta$.

Following \cite{GO17} we note that
$\eta_{ij} =\yij-\xij^\tran\beta = \ai+\bj+\eij$
has the same parameter $\theta$ as $Y_{ij}$ have.
We then take a consistent estimate of $\beta$,
in this case $\hat\beta_{\ols}$ that \cite{GO17} show is consistent for $\beta$,
and define  $\hat\eta_{ij}  =\yij-\xij^\tran\hat\beta_\ols$.
We then estimate $\theta$ by the above method
after replacing $\yij$ by $\hat\eta_{ij}$.
For the Stitch Fix data we obtained
$\hat{\sigma}_{A}^{2} = 1.14$ (customers),
$\hat{\sigma}^{2}_{B} = 0.11$ (items) 
and $\hat{\sigma}^{2}_{E} = 4.47$. 

\subsection{Computing $\hat\beta_\gls$}\label{sec:wholeshebang}

The estimated coefficients $\hat\beta_\gls$ and their standard errors are presented in a table in the appendix.
Open-source R code at
\url{https://github.com/G28Sw/backfit_code}
does these computations.
Here is a concise description of the algorithm we used:
\begin{compactenum}[\quad 1)]
\item Compute $\hat\beta_\ols$ via \eqref{eq:bhatols}.
\item Get residuals $\hat\eta_{ij} =\yij -\xij^\tran\hat\beta_{\ols}$.
\item Compute $\ssah$, $\ssbh$ and $\sseh$ by the method of moments on $\hat\eta_{ij}$.
\item Compute $\wt\cx=(I_N-\wt\cS_G)\cx$ using doubly centered backfitting $\mthree$.
\item Compute $\hat\beta_{\gls}$ by~\eqref{eq:covbhatgls}.
\item If we want BLUPs  $\hat\bsa$ and $\hat\bsb$ backfit
$\cy -\cx\hat\beta_{\gls}$ to get them.
\item Compute $\wh\cov(\hat\beta_{\gls})$ by plugging
$\ssah$, $\ssbh$ and $\sseh$ into $\cv$ at~\eqref{eq:covbhatgls}.
\end{compactenum}
\smallskip

Stage $k$ of backfitting provides $(\tilde\cS_G\cx)^{(k)}$.
We iterate until
$$
\frac{\Vert (\tilde\cS_G\cx)^{(k+1)}-(\tilde\cS_G\cx)^{(k)}\Vert^2_F}{\Vert (\tilde\cS_G\cx)^{(k)}\Vert^2_F}
< \epsilon
$$
where $\Vert \cdot \Vert_F$ is the Frobenius norm
(root mean square of all elements).
Our numerical results use $\epsilon =10^{-8}$.

{
When we want $\wh\cov(\hat\beta_{\gls})$ then we need
to use a backfitting strategy with a symmetric smoother
$\tilde\cS_G$.  This holds for $\mzero$, $\mtwo$ and $\mthree$
but not $\mone$.
After computing $\hat\beta_{\gls}$ one can return to step 2,
form new residuals
 $\hat\eta_{ij} =\yij -\xij^\tran\hat\beta_{\gls}$
and continue through steps 3--7.
We have seen small differences from doing this.
}  

\subsection{Quantifying inefficiency and naivete of OLS}
In the introduction we mentioned two serious problems with the use of OLS on crossed
random effects data. The first is that OLS is naive about correlations in the
data and this can lead it to severely underestimate the variance of $\hat\beta$.
The second is that OLS is inefficient compared to GLS by the Gauss-Markov theorem.
Let $\hat\beta_\ols$ and $\hat\beta_\gls$ be the OLS and GLS
estimates of $\beta$, respectively. We can compute their
corresponding variance estimates
$\wh\cov_\ols(\hat\beta_\ols)$ and $\wh\cov_\gls(\hat\beta_\gls)$.
We can also find
$\wh\cov_\gls(\hat\beta_\ols)$, the variance under our GLS model of the
linear combination of $\yij$ values that  OLS uses.
This section explore them graphically.

We can quantify the naivete of OLS
via the ratios
$\wh\cov_{\gls}(\hat\beta_{\ols,j})/\wh\cov_{\ols}(\hat\beta_{\ols,j})$
for $j=1,\dots,p$.
Figure~\ref{fig:OLSisnaive} plots these values. They range from $ 1.75$
to $345.28$ and can be interpreted as factors by which OLS naively overestimates
its sample size.
The largest and second largest ratios are for material indicators
corresponding to `Modal' and `Tencel', respectively. These appear
to be two names for the same product with Tencel being a trademarked name
for Modal fibers (made from wood).
We can also identify the linear combination of $\hat\beta_\ols$
for which $\ols$ is most naive. We maximize
the ratio
$x^\tran\wh\cov_{\gls}(\hat\beta_{\ols})x/x^\tran\wh\cov_{\ols}(\hat\beta_{\ols})x$
over $x\ne0$.
The resulting maximal ratio is the largest eigenvalue of
$$\wh\cov_{\ols}(\hat\beta_{\ols}) ^{-1}
\wh\cov_{\gls}(\hat\beta_{\ols})$$
and it is about $361$ for the Stitch Fix data.

\begin{figure}
\centering
\includegraphics[width=.9\hsize]{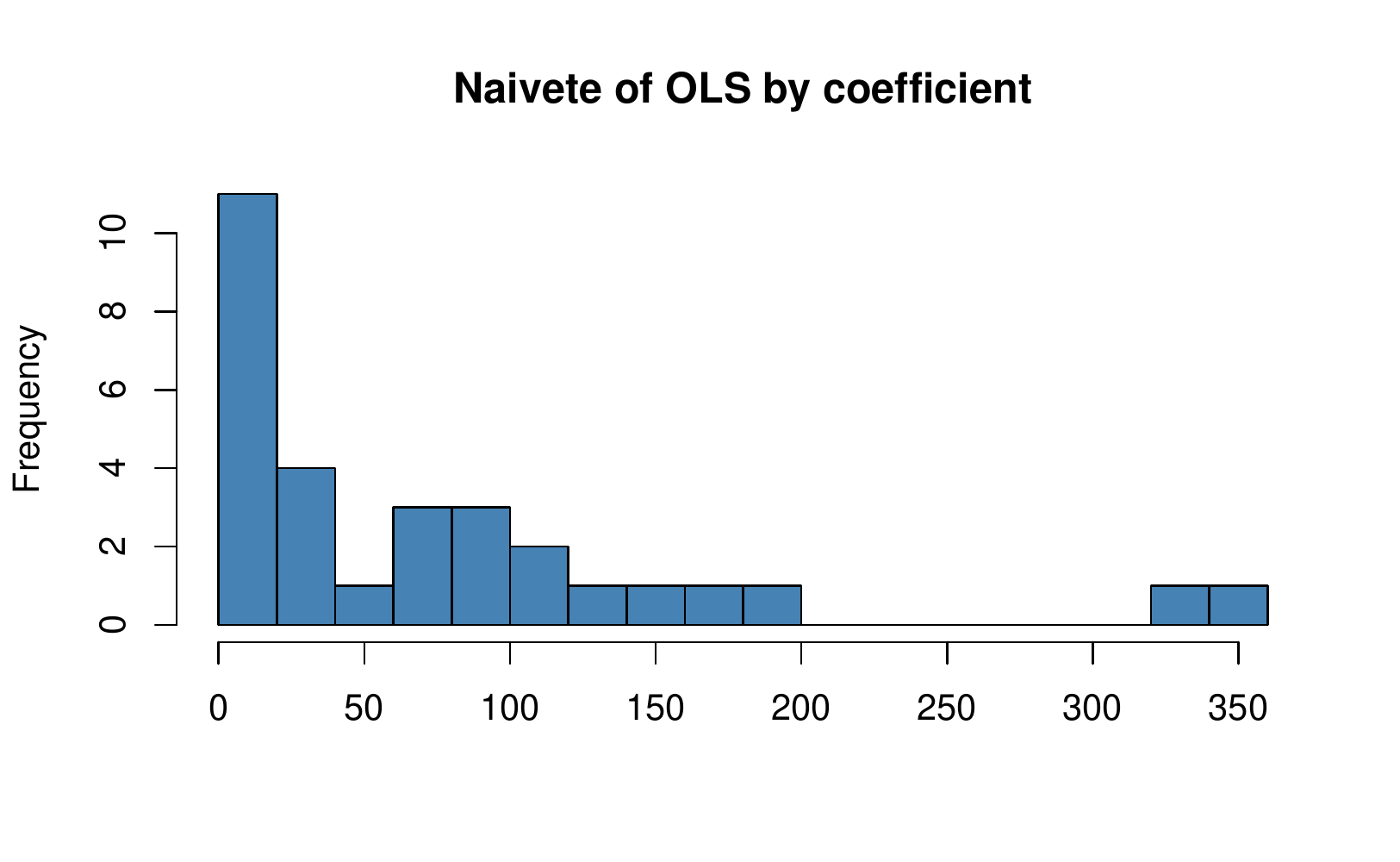}
\caption{\label{fig:OLSisnaive}
OLS naivete
$\wh\cov_{\gls}(\hat\beta_{\ols,j})/\wh\cov_{\ols}(\hat\beta_{\ols,j})$
for coefficients $\beta_j$ in the Stitch Fix data.
}
\end{figure}

We can quantify the inefficiency of OLS
via the ratio
$\wh\cov_{\gls}(\hat\beta_{\ols,j})/\wh\cov_{\gls}(\hat\beta_{\gls,j})$
for $j=1,\dots,p$.
Figure~\ref{fig:OLSisinefficient} plots these values. They range from just over $1$
to $50.6$ and can be interpreted as factors by which using
OLS reduces the effective sample size.  There is a clear outlier: the coefficient of the match
variable is very inefficiently estimated by OLS.  The second largest inefficiency
factor is for the intercept term.
The most inefficient linear combination of $\hat\beta$ reaches a
variance ratio of $52.6$, only slightly more inefficient than the match coefficient alone.

\begin{figure}
\centering
\includegraphics[width=.9\hsize]{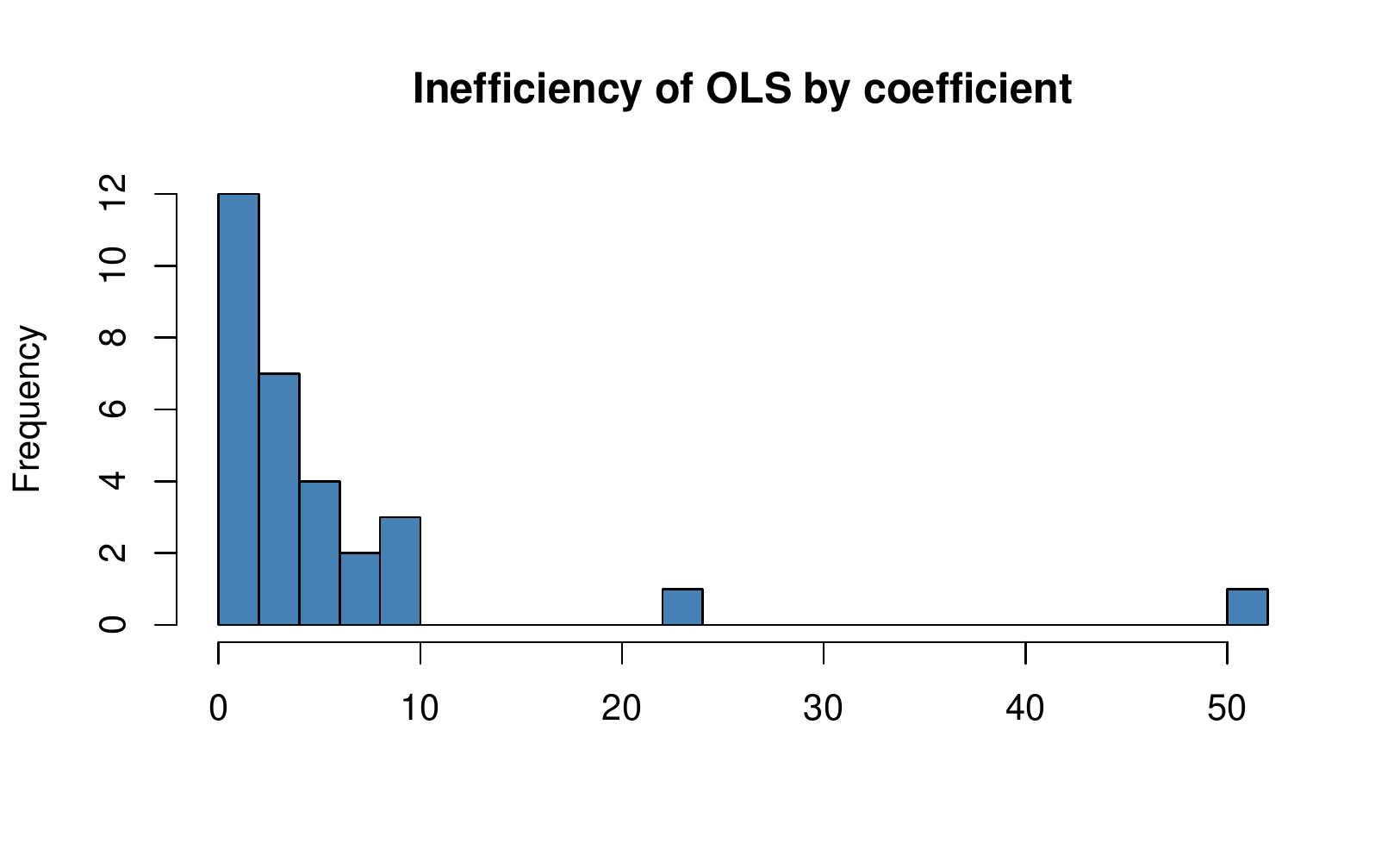}
\caption{\label{fig:OLSisinefficient}
OLS inefficiency
$\wh\cov_{\gls}(\hat\beta_{\ols,j})/\wh\cov_{\gls}(\hat\beta_{\gls,j})$
for coefficients $\beta_j$ in the Stitch Fix data.
}
\end{figure}

The variables for which OLS is more naive tend to also be the variables for
which it is most inefficient. Figure~\ref{fig:naivevsinefficient} plots these
quantities against each other for the $30$ coefficients in our model.

\begin{figure}[t]
\centering
\includegraphics[width=.8\hsize]{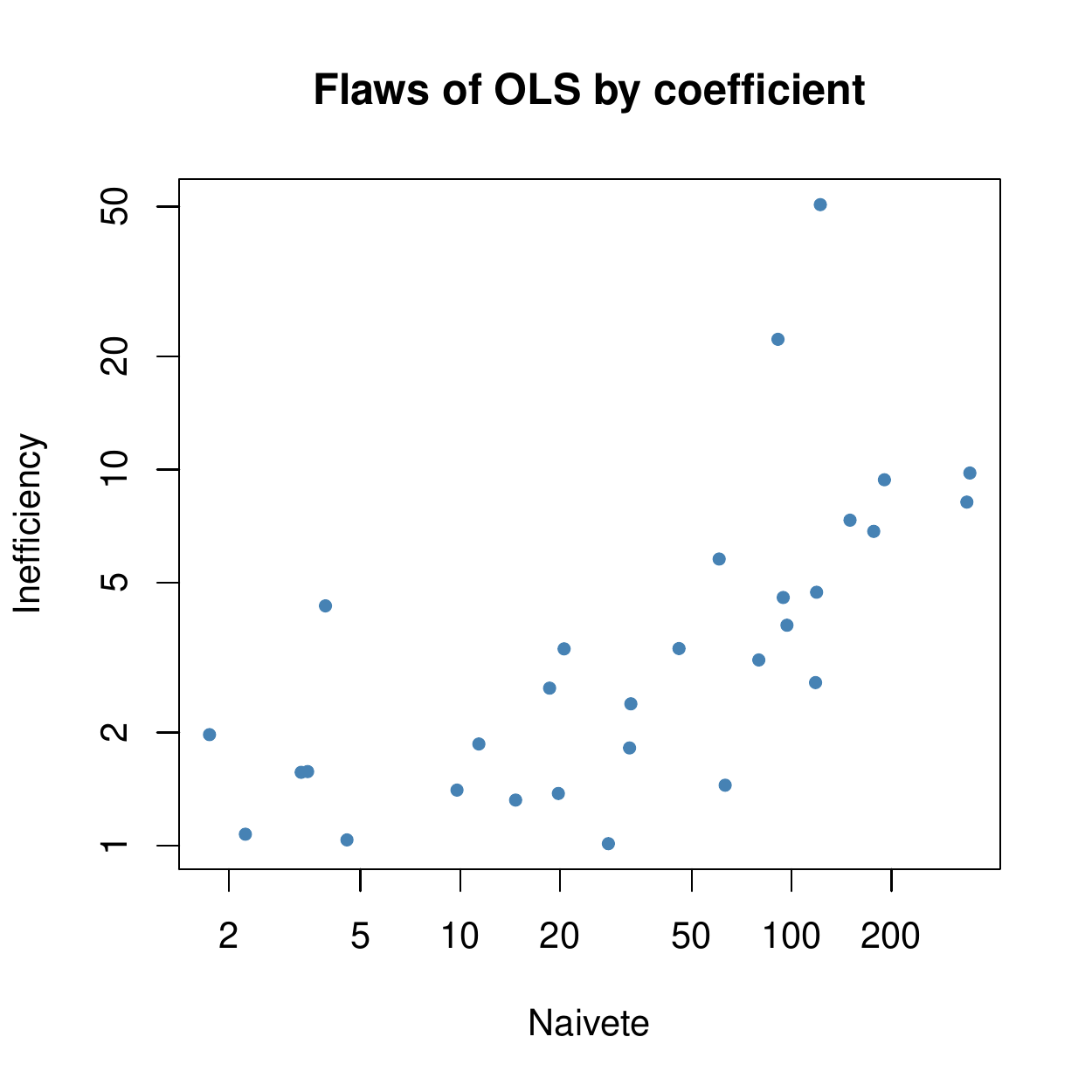}
\caption{\label{fig:naivevsinefficient}
Inefficiency vs naivete for OLS coefficients in the Stitch Fix data.
}
\end{figure}

\subsection{Convergence speed of backfitting}

The Stitch Fix data have row and column sample sizes
that are much more uneven than our sampling model for $Z$ allows.
Accordingly we cannot rely on Theorem~\ref{thm:m1norm1} to show that
backfitting must converge rapidly for it.

The sufficient conditions in that theorem may not be necessary
and we can compute
our norms and the spectral radius on
the update matrices for  the Stitch Fix data using some sparse matrix computations.
Here $Z\in\{0,1\}^{762,752\times6318}$,
so $M^{(k)}\in\real^{6318\times 6318}$ for $k \in \lbrace0,1,2,3\rbrace$.
The results are
$$
\begin{pmatrix}
\Vert \mzero\Vert_1 \ & \  \Vert \mzero\Vert_2 \ & \  |\lambda_{\max}(\mzero)|\\[.25ex]
\Vert \mone\Vert_1 \ & \  \Vert \mone\Vert_2 \ & \  |\lambda_{\max}(\mone)|\\[.25ex]

\Vert \mtwo\Vert_1 \ & \  \Vert \mtwo\Vert_2 \ & \  |\lambda_{\max}(\mtwo)|\\[.25ex]
\Vert \mthree\Vert_1 \ & \  \Vert \mthree\Vert_2 \ & \  |\lambda_{\max}(\mthree)|
\end{pmatrix}
=\begin{pmatrix}
31.9525 \ & \ 1.4051 \ & \  0.64027 \\[.75ex]
11.2191 \ & \  0.4512 \ & \  0.33386\\[.75ex]
\phz8.9178 \ & \  0.4541 \ & \  0.33407\\[.75ex]
\phz9.2143\ & \ 0.4546 & \ 0.33377\\
\end{pmatrix}.
$$
All the updates have spectral radius comfortably below one.
The centered updates have $L_2$ norm below one
but the uncentered update does not.
Their $L_2$ norms are somewhat larger than their spectral
radii because those matrices are not quite symmetric.
The two largest eigenvalue moduli for $\mzero$ are $0.6403$ and $0.3337$
and the centered updates have spectral radii close to the second
largest eigenvalue of $\mzero$.
This is consistent with an intuitive explanation that the space spanned
by a column of $N$ ones that is common to the columns spaces
of $\cz_A$ and $\cz_B$ is the {biggest impediment} to $\mzero$ and that
all three centering strategies essentially remove it.
The best spectral radius is for $\mthree$, which employs two principled
centerings, although in this data set it made little difference.
Our backfitting algorithm took $8$ iterations when applied to $\cx$
and $12$ more to compute the BLUPs.
We used a convergence threshold of $10^{-8}.$

\section{Discussion}\label{sec:discussion}

We have shown that the cost of our backfitting algorithm
is $O(N)$ under strict conditions that are nonetheless
much more general than having $\nid = N/C$
for all $i=1,\dots,R$ and  $\ndj = N/R$ for all $j=1,\dots,C$
as in \cite{papa:robe:zane:2020}.
As in their setting, the backfitting algorithm scales empirically to
much more general problems than those for which
rapid convergence can be proved.
Our contour map of the spectral radius of the update
matrix $M$ shows that this norm is well below $1$
over many more $(\rho,\kappa)$ pairs that our
theorem covers. The difficulty in extending our
approach to those settings is that the spectral radius
is a much more complicated function of the observation
matrix $Z$ than the $L_1$ norm is.

Theorem 4 of \cite{papa:robe:zane:2020}
has the rate of convergence for their collapsed Gibbs
sampler for balanced data.
It involves an auxilliary convergence rate $\rhoaux$
defined as follows.
Consider the Gibbs sampler on $(i,j)$ pairs where
given $i$ a random $j$ is chosen with probability $\zij/\nid$
and given $j$ a random $i$ is chosen with probability
$\zij/\ndj$.  That Markov chain has invariant distribution $\zij/N$
on $(i,j)$ pairs and $\rhoaux$ is the rate at which the chain converges.
In our notation
$$
\rhoprz = \frac{N\ssa}{N\ssa+R\sse}\times\frac{N\ssb}{N\ssb+C\sse}\times\rhoaux.
$$
In sparse data $\rhoprz\approx\rhoaux$ and under our asymptotic
setting $|\rhoaux-\rhoprz|\to0$.
\cite{papa:robe:zane:2020} remark that $\rhoaux$ tends to decrease
as the amount of data increases. When it does, then their algorithm
takes $O(1)$ iterations and costs $O(N)$.
They explain that $\rhoaux$ should decrease as the data set
grows because the auxiliary process then gets greater connectivity.
That connectivity increases for bounded $R$ and $C$ with increasing $N$
and from their notation, allowing multiple observations
per $(i,j)$ pair it seems like they have this sort of infill
asymptote in mind.
For sparse data from electronic commerce we think that
an asymptote like the one we study where $R$, $C$ and $N$
all grow is a better description.
It would be interesting to see how $\rhoaux$ develops under such a model.

In Section 5.3 \cite{papa:robe:zane:2020}
state that the convergence rate of the collapsed Gibbs sampler
is $O(1)$ regardless of the asymptotic regime.  That section is about
a more stringent `balanced cells' condition where every $(i,j)$ combination
is observed the same number of times, so it does not describe
the `balanced levels' setting where $\nid=N/R$ and $\ndj=N/C$.
Indeed they provide a counterexample in which there are two
disjoint communities of users and two disjoint sets of items
and each user in the first community has rated every item
in the first item set (and no others) while each user in the
second community has rated every item in the second item
set (and no others).  That configuration leads to an unbounded mixing time
for collapsed Gibbs.  It is also one where backfitting takes
an increasing number of iterations as the sample size grows.

There are interesting parallels between methods to sample a high
dimensional Gaussian distribution with covariance matrix $\Sigma$
and iterative solvers for the system $\Sigma \bsx = \bsb$.
See \cite{good:soka:1989} and \cite{RS97}
for more on how the convergence rates
for these two problems coincide.
We found that backfitting with one or both updates centered
worked much better than uncentered backfitting.
\cite{papa:robe:zane:2020} used a collapsed sampler
that analytically integrated out the global mean of their model in each update
of a block of random effects.

Our approach treats $\ssa$, $\ssb$ and $\sse$ as nuisance parameters.
We plug in a consistent method of moments based estimator of them
in order to focus on the backfitting iterations.
In Bayesian computations, maximum a posteriori estimators of
variance components under non-informative priors can be
problematic for hierarchical models \cite{gelm:2006},
and so perhaps maximum likelihood estimation of these
variance components would also have been challenging.

Whether one prefers a GLS estimate or a Bayesian one
depends on context and goals.   We believe that there is a strong
computational advantage to GLS for large data sets.
The cost of one backfitting iteration is comparable to the cost to generate
one more sample in the MCMC.  We may well find that only a dozen
or so iterations are required for convergence of the GLS.  A Bayesian
analysis requires a much larger number of draws from the posterior
distribution than that.
For instance, \cite{gelm:shir:2011} recommend an effective sample size of about $100$
posterior draws, with autocorrelations requiring a larger actual sample size.
\cite{vats:fleg:jone:2019} advocate even greater effective sample sizes.

It is usually reasonable to assume that there is a selection
bias underlying which data points are observed.
Accounting for any such selection bias must necessarily
involve using information or assumptions from outside the data set at
hand.  We expect that any approach to take proper account of
informative missingness must also make use of solutions to
GLS perhaps after reweighting the observations.
Before one develops any such methods, it is necessary
to first be able to solve GLS without regard to missingness.

Many of the problems in electronic commerce involve categorical outcomes,
especially binary ones, such as whether an item was purchased or not.
Generalized linear mixed models are then appropriate ways to handle
crossed random effects, and we expect that the progress made here
will be useful for those problems.

\section*{Acknowledgements}
This work was supported by the U.S.\ National Science Foundation under grant IIS-1837931.
We are grateful to Brad Klingenberg and Stitch Fix for sharing some test data with us.
We thank the reviewers for remarks that have helped us improve the paper.

\bibliographystyle{imsart-nameyear} 
\bibliography{bigdata}       

\begin{thebibliography}{}

\bibitem[Bates et~al., 2015]{lme4}
Bates, D., M{\"a}chler, M., Bolker, B., and Walker, S. (2015).
\newblock Fitting linear mixed-effects models using {lme4}.
\newblock {\em Journal of Statistical Software}, 67(1):1--48.

\bibitem[Buja et~al., 1989]{buja:hast:tibs:1989}
Buja, A., Hastie, T., and Tibshirani, R. (1989).
\newblock Linear smoothers and additive models (with discussion).
\newblock {\em The Annals of Statistics}, pages 453--510.

\bibitem[Cameron et~al., 2011]{came:gelb:mill:2011}
Cameron, A.~C., Gelbach, J.~B., and Miller, D.~L. (2011).
\newblock Robust inference with multiway clustering.
\newblock {\em Journal of Business \& Economic Statistics}, 29(2):238--249.

\bibitem[Gao, 2017]{gao:thesis}
Gao, K. (2017).
\newblock {\em Scalable Estimation and Inference for Massive Linear Mixed
  Models with Crossed Random Effects}.
\newblock PhD thesis, Stanford University.

\bibitem[Gao and Owen, 2017]{GO17}
Gao, K. and Owen, A.~B. (2017).
\newblock Efficient moment calculations for variance components in large
  unbalanced crossed random effects models.
\newblock {\em Electronic Journal of Statistics}, 11(1):1235--1296.

\bibitem[Gao and Owen, 2019]{crelin}
Gao, K. and Owen, A.~B. (2019).
\newblock Estimation and inference for very large linear mixed effects models.
\newblock {\em Statistica Sinica}.
\newblock To appear.

\bibitem[Gelman, 2006]{gelm:2006}
Gelman, A. (2006).
\newblock Prior distributions for variance parameters in hierarchical models
  (comment on article by {Browne and Draper}).
\newblock {\em Bayesian analysis}, 1(3):515--534.

\bibitem[Gelman and Hill, 2006]{gelm:hill:2006}
Gelman, A. and Hill, J. (2006).
\newblock {\em Data analysis using regression and multilevel/hierarchical
  models}.
\newblock Cambridge University Press, Cambridge.

\bibitem[Gelman and Shirley, 2011]{gelm:shir:2011}
Gelman, A. and Shirley, K. (2011).
\newblock Inference from simulations and monitoring convergence.
\newblock In Brooks, S., Gelman, A., Jones, G., and Meng, X.-L., editors, {\em
  Handbook of {Markov chain Monte Carlo}}, volume~6, pages 163--174. CRC Press
  Boca Raton, FL.

\bibitem[Goodman and Sokal, 1989]{good:soka:1989}
Goodman, J. and Sokal, A.~D. (1989).
\newblock Multigrid {Monte Carlo} method. {Conceptual foundations}.
\newblock {\em Physical Review D}, 40(6):2035--2071.

\bibitem[Hastie and Tibshirani, 1990]{hast:tibs:1990}
Hastie, T. and Tibshirani, R.~J. (1990).
\newblock {\em Generalized Additive Models}.
\newblock Chapman and Hall, Boca Raton, FL.

\bibitem[Owen, 2007]{pbs}
Owen, A.~B. (2007).
\newblock The pigeonhole bootstrap.
\newblock {\em The Annals of Applied Statistics}, 1(2):386--411.

\bibitem[Papaspiliopoulos et~al., 2020]{papa:robe:zane:2020}
Papaspiliopoulos, O., Roberts, G.~O., and Zanella, G. (2020).
\newblock Scalable inference for crossed random effects models.
\newblock {\em Biometrika}, 107(1):25--40.

\bibitem[{R Core Team}, 2015]{R:lang:2015}
{R Core Team} (2015).
\newblock {\em R: A Language and Environment for Statistical Computing}.
\newblock R Foundation for Statistical Computing, Vienna, Austria.

\bibitem[Roberts and Sahu, 1997]{RS97}
Roberts, G.~O. and Sahu, S.~K. (1997).
\newblock Updating schemes, correlation structure, blocking and
  parameterization for the {Gibbs} sampler.
\newblock {\em Journal of the Royal Statistical Society, Series B}, pages
  291--317.

\bibitem[Robinson, 1991]{robinson91:_that_blup}
Robinson, G. (1991).
\newblock That {BLUP} is a good thing: the estimation of random effects.
\newblock {\em Statistical Science}, 6(1):15--51.

\bibitem[Searle et~al., 1992]{sear:case:mccu:1992}
Searle, S.~R., Casella, G., and McCulloch, C.~E. (1992).
\newblock {\em Variance Components}.
\newblock Wiley, New York.

\bibitem[Strassen, 1969]{stra:1969}
Strassen, V. (1969).
\newblock Gaussian elimination is not optimal.
\newblock {\em Numerische mathematik}, 13(4):354--356.

\bibitem[Vats et~al., 2019]{vats:fleg:jone:2019}
Vats, D., Flegal, J.~M., and Jones, G.~L. (2019).
\newblock Multivariate output analysis for {Markov chain Monte Carlo}.
\newblock {\em Biometrika}, 106(2):321--337.

\end{thebibliography}

\vfill\eject
\section{Appendix}

Table~\ref{tab:Stitch_fix_regression} shows results of $\ols$ and $\gls$
regression for the Stitch Fix data in Section~\ref{sec:stitch}.
OLS is estimated to be naive when
$\wh{\mathrm{SE}}_{\ols}(\hat\beta_\ols)<\wh{\mathrm{SE}}_{\gls}(\hat\beta_\ols)$
and inefficient when
$\wh{\mathrm{SE}}_{\ols}(\hat\beta_\ols)>\wh{\mathrm{SE}}_{\gls}(\hat\beta_\gls)$.
Estimates that are more than double their corresponding standard error
get an asterisk.
{\color{red}

\newcolumntype{d}[1]{D{.}{.}{#1}}
\begin{table}[h!]
\centering
\caption{\label{tab:Stitch_fix_regression}
Stitch Fix
Regression Results}

\begin{tabular}{l d{3.7} d{3.7} d{3.7} d{3.7} d{3.7}}
\toprule
& \multicolumn{1}{c}{$\hat{\beta}_{\ols}$} & \multicolumn{1}{c}{$\widehat{\mathrm{SE}}_{\ols}(\hat{\beta}_{\ols})$} & \multicolumn{1}{c}{$\widehat{\mathrm{SE}}_{\gls}(\hat{\beta}_{\ols})$} & \multicolumn{1}{c}{$\hat{\beta}_{\gls}$} & \multicolumn{1}{c}{$\widehat{\mathrm{SE}}_{\gls}(\hat{\beta}_{\gls})$} \\
\midrule
Intercept& 4.635^{*} & 0.005397 & 0.05148 & 5.103^{*} & 0.01092\\
Match & 5.048^{*} & 0.01174 & 0.1297 & 3.442^{*} & 0.01823 \\
 $\mathbb{I}\{\text { client edgy }\}$ & 0.001020 & 0.002443 & 0.004444 & 0.003041 & 0.003550 \\
 $\mathbb{I}\{\text { item edgy }\}$ & -0.3358^{*} & 0.004253 & 0.03307 & -0.3515^{*} & 0.01375  \\
 $\mathbb{I}\{\text { client edgy }\}$ & & & \\
 $* \mathbb{I}\{\text { item edgy }\}$ & 0.3925^{*} & 0.006229 & 0.01233 & 0.3793^{*} & 0.005916\\
 $\mathbb{I}\{\text { client boho }\}$ & 0.1386^{*} & 0.002264 & 0.004211 & 0.1296^{*}& 0.003356\\
 $\mathbb{I}\{\text { item boho }\}$ & -0.5499^{*} & 0.005981 & 0.02713 &-0.6266^{*}& 0.01485  \\
 $\mathbb{I}\{\text { client boho }\}$ & & & \\
 $* \mathbb{I}\{\text { item boho }\}$ & 0.3822^{*} & 0.007566 &  0.01001 & 0.3763^{*} & 0.007123 \\
 Acrylic & -0.06482^{*} & 0.003778 & 0.03371 &-0.005360 & 0.01909                      \\
 Angora & -0.01262 & 0.007848 &  0.08530 &  0.07486 & 0.05177\\
 Bamboo & -0.04593 & 0.06215 & 0.2096 & 0.03251& 0.1535      \\
 Cashmere & -0.1955^{*} & 0.02484 & 0.1414& 0.008930   & 0.1048               \\
 Cotton & 0.1752^{*} & 0.003172 & 0.04220 & 0.1033^{*} &0.01612\\
 Cupro & 0.5979^{*} & 0.3016 &  0.4519&  0.2089 & 0.4363\\
 Faux Fur & 0.2759^{*} & 0.02008 & 0.07694 & 0.2749^{*} & 0.06691                 \\
 Fur & -0.2021^{*} & 0.03121 & 0.1388 & -0.07924 & 0.1182    \\
 Leather & 0.2677^{*} & 0.02482 & 0.07759 &0.1674^{*}&0.06545 \\
 Linen & -0.3844^{*} & 0.05632 & 0.2429 & -0.08658&0.1499 \\
 Modal & 0.002587 & 0.009775 & 0.1816&0.1388^{*} &0.05804                         \\
 Nylon & 0.03349^{*} & 0.01552 & 0.08878 &0.08174&0.05751\\
 Patent Leather & -0.2359 & 0.1800 & 0.3838 & -0.3764&0.3771   \\
 Pleather & 0.4163^{*} & 0.008916 & 0.08774 & 0.3292^{*}&0.04468 \\
 PU & 0.4160^{*} & 0.008225 & 0.07989 & 0.4579^{*}&0.03737\\
 PVC & 0.6574^{*} & 0.06545 & 0.3462 & 0.9688^{*}&0.3441 \\
 Rayon & -0.01109^{*} & 0.002951 & 0.04074 & 0.05155^{*}&0.01329  \\
 Silk & -0.1422^{*} & 0.01317 & 0.08907 & -0.1828^{*}&0.04871 \\
 Spandex & -0.3916^{*} & 0.00931 &0.1373  & 0.4140^{*} &0.1141  \\
 Tencel & 0.4966^{*} & 0.01729 & 0.1712&  0.1234^{*}&0.05982 \\
 Viscose & 0.04066^{*} & 0.006953 & 0.08519 & -0.02259&0.03145   \\
Wool & -0.06021^{*} & 0.006611 & 0.07211&-0.05883&0.03319\\
\bottomrule
\end{tabular}

\end{table}

\end{document}